\documentclass[]{cmslatex}
\usepackage[paperwidth=7in, paperheight=10in, margin=.875in]{geometry}
\usepackage{graphics}                
\usepackage{url}
\usepackage{bm,booktabs}
\usepackage{amsfonts,amssymb}
\usepackage[ruled]{algorithm2e}
\usepackage{mathrsfs}
\usepackage{algorithmic}
\usepackage{subcaption}
\usepackage{graphicx,epsfig}
\usepackage{amsmath}
\usepackage[all]{xy}
\usepackage{caption,fixltx2e}
\usepackage[flushleft]{threeparttable}  
\usepackage{bbold}
\providecommand{\keyword}[1]{\textbf{Keywords: } #1}

\newcommand{\nit}{\noindent}

\newcommand{\be}{\begin{equation}}
\newcommand{\ee}{\end{equation}}
\newcommand{\ba}{\begin{eqnarray*}}
\newcommand{\ea}{\end{eqnarray*}}
\newcommand{\bi}{\begin{itemize}}
\newcommand{\ei}{\end{itemize}}

\newcommand{\eps}{\mbox{$\epsilon$}}

\newtheorem{theo}{\textbf{Theorem}}[section]
\newtheorem{defin}{\textbf{Definition}}[section]

\newtheorem{lem}{\textbf{Lemma}}[section]

\newtheorem{rmk}{Remark}[section]
\newtheorem{cond}{\textbf{Condition}}
\renewcommand{\theequation}{\arabic{section}.\arabic{equation}}
\newcommand{\comments}[1]{}

\graphicspath{{pic/}}

\begin{document}
\title{Minimization of Transformed $L_1$ Penalty: Theory, Difference of Convex Function Algorithm, 
and Robust Application in Compressed Sensing}

\author{Shuai~Zhang,
        and~Jack~Xin
\thanks{
S. Zhang and J. Xin were partially supported by NSF grants DMS-0928427,
DMS-1222507, and DMS-1522383. They are with the Department of Mathematics, University of California, Irvine, CA, 92697, USA.
E-mail: szhang3@uci.edu; jxin@math.uci.edu.
}}

\maketitle

\begin{abstract}
We study the minimization problem of a non-convex sparsity promoting penalty function, 
the transformed $l_1$ (TL1), and its application in compressed sensing (CS). The TL1 penalty 
interpolates $l_0$ and $l_1$ norms through a nonnegative parameter $a \in (0,+\infty)$,  
similar to $l_p$ with $p \in (0,1]$, and is known to satisfy   
unbiasedness, sparsity and Lipschitz continuity properties. 
We first consider the constrained minimization problem, and discuss the 
exact recovery of $l_0$ norm minimal solution based on the null space property (NSP).
We then prove the stable recovery of $l_0$ norm minimal solution if 
the sensing matrix $A$ satisfies a restricted 
isometry property (RIP).  We formulated a normalized problem to overcome the 
lack of scaling property of the TL1 penalty function. 
For a general sensing matrix $A$, we show that the support set of a local minimizer 
corresponds to linearly independent columns of $A$. 
Next, we present difference of convex algorithms for 
TL1 (DCATL1) in computing TL1-regularized constrained and 
unconstrained problems in CS. 
The DCATL1 algorithm involves outer and inner loops of iterations, one time matrix inversion, 
repeated shrinkage operations and matrix-vector multiplications. 
The inner loop concerns an $l_1$ minimization problem on which we employ 
the Alternating Direction Method of Multipliers (ADMM). 
For the unconstrained problem, we prove convergence of 
DCATL1 to a stationary point satisfying 
the first order optimality condition. 
In numerical experiments, we identify the optimal value $a=1$, 
and compare DCATL1 with other CS algorithms on two classes of 
sensing matrices: Gaussian random matrices and over-sampled discrete 
cosine transform matrices (DCT). 
Among existing algorithms, the iterated reweighted least squares method based on $l_{1/2}$ norm 
is the best in sparse recovery for Gaussian matrices, 
and the DCA algorithm based on $l_1$ minus $l_2$ 
penalty is the best for over-sampled DCT matrices. 
We find that for both classes of sensing matrices, the 
performance of DCATL1 algorithm (initiated with $l_1$ minimization) 
always ranks near the top (if not the top), and 
is the {\it most robust choice} insensitive to the 
conditioning of the sensing matrix $A$. 
DCATL1 is also competitive in comparison with DCA 
on other non-convex penalty functions commonly used in statistics 
with two hyperparameters.
\end{abstract}

\keyword{
Transformed $l_1$ penalty, sparse signal recovery theory, difference of convex function algorithm, 
convergence analysis, coherent random matrices, compressed sensing, robust recovery. }


\section{Introduction}\label{intro}
Compressed sensing \cite{candes2006stable,Don_06} has generated enormous interest and research 
activities in mathematics, statistics, information science and elsewhere. 
A basic problem is to reconstruct a sparse signal under a few linear 
measurements (linear constraints) far less than the dimension of the ambient space of the signal. 
Consider a sparse signal $x \in \Re^N $, an $M\times N$ sensing matrix $A$ and an observation 
$y \in \Re^M$, $M \ll N$, such that: $	y = A\, x + \epsilon, $
where $\epsilon$ is an $N$-dimensional observation error. The main objective 
is to recover $x$ from $y$. 
\medskip

The direct approach is $l_0$ optimization in either a constrained formulation: 
\begin{equation}\label{eq:l0 cons}
\min \limits_{x \in \Re^N} \| x \| _0, \;\; s.t. \;\;\  A x=y,
\end{equation}
or an unconstrained $l_0$ regularized optimization: 
\begin{equation} \label{eq:l0 uncons}
\min \limits_{x \in \Re^N} \{  \| y - A\, x \|_2^2  + \lambda \|x\|_0 \}
\end{equation}
with a positive regularization parameter $\lambda$. Since minimizing $l_0$ norm is NP-hard 
\cite{L0-NP-natarajan1995sparse}, many viable alternatives have been sought. Greedy methods 
(matching pursuit \cite{MP-mallat1993matching}, othogonal matching pursuits (OMP) 
\cite{OMP-tropp2005signal}, and regularized OMP (ROMP) \cite{ROMP-needell2010signal}) work 
well if the dimension $N$ is not too large. For the unconstrained problem (\ref{eq:l0 uncons}), 
the penalty decomposition method \cite{penalty-decomp} replaces the term $\lambda \| x \|_0$ 
by $\rho_k \|x - z\|_{2}^{2} + \lambda \| z \|_0$, and minimizes over $(x, z)$ for a diverging 
sequence $\rho_k$. The variable $z$ allows the iterative hard thresholding procedure. 
\medskip

The relaxation approach is to replace $l_0$ norm by a continuous sparsity promoting penalty 
function $P(x)$. Convex relaxation uniquely selects $P(\cdot)$ as the $l_1$ norm. The 
resulting problem is known as basis pursuit (LASSO in the over-determined regime \cite{lasso}). 
The $l_1$ algorithms include $l_1$-magic \cite{candes2006stable}, Bregman and split Bregman 
methods \cite{splitBregman, Bregman:yin2008} and yall1 \cite{yall1}. 
Theoretically, Cand\`es, Tao and coauthors  
introduced the restricted isometry property (RIP) on $A$  
to establish the equivalent and unique global solution to $l_0$ 
minimization and stable sparse recovery results 
\cite{candes2005decoding, candes2005error, candes2006stable}. 
\medskip

There are many choices of $P(\cdot)$ for non-convex relaxation, almost all of them are known in 
statistics \cite{fan2001variable,spnet_11}. One is the $l_p$ norm (a.k.a. bridge penalty)  
($p \in (0,1)$) with known $l_0$ equivalence under RIP \cite{chartrand2007nonconvex}. 
The $l_{1/2}$ norm is representative of this class of penlty functions, with the associated  
reweighted least squares and half-thresholding algorithms 
for computation \cite{reweighted-l1/2-WTYing-2013improved,xian-half, xian-half-sa}. 
Near the RIP regime, 
$l_{1/2}$ penalty tends to have higher success rate of sparse reconstruction than $l_1$. 
However, it is not as good as $l_1$ if the sensing matrix $A$ is 
far from RIP \cite{l1-l2-lou2014computing,l1-l2-yinminimization} 
as we shall see later as well. In the highly non-RIP (coherent) regime, 
it is recently found that the difference of 
$l_1$ and $l_2$ norm minimization gives the best sparse 
recovery results \cite{l1-l2-yinminimization,l1-l2-lou2014computing}. 
{\it It is therefore of both theoretical and practical interest to find a 
non-convex penalty that is consistently 
better than $l_1$ and always ranks among the top in sparse recovery 
whether the sensing matrix satisfies RIP or not.} 
\medskip

The $l_p$ penalty functions are however known to bias towards large peaks.  
In the statistics literature of variable selection, 
Fan and Li \cite{fan2001variable} advocated for classes 
of penalty functions with three desired properties: 
{\it unbiasedness, sparsity} and {\it continuity}. 
To help identify such a penalty function denoted by $\rho(\cdot)$, 
Fan and Lv \cite{transformed-l1} proposed 
the following condition for characterizing unbiasedness 
and sparsity promoting properties.
\begin{cond}
The penalty function $ \rho(\cdot) $ satisfies:
\begin{enumerate}
	\item[(i)] $\rho(t)$ is increasing and concave in $t \in [0,\infty)$, 
    \item[(ii)] $\rho'(t)$ is continuous with $\rho'(0+) \in(0, \infty)$.
\end{enumerate}
\end{cond}
It follows that $\rho'(t)$ is positive and decreasing, and $\rho'(0+)$ is the upper bound of $\rho'(t)$.
The penalties satisfying Condition 1 and $\lim_{t \to \infty} \rho^{'}(t) = 0 $ 
enjoy both unbiasedness and sparsity \cite{transformed-l1}. 
Though continuity does not generally hold for 
this class of penalty functions, a special one parameter family of Lipschitz 
continuous functions, the so called 
transformed $l_1$ functions \cite{original TL1}, satisfy all three desired properties above \cite{transformed-l1}. 
\medskip

In this paper, we show that {\bf minimizing the non-convex  
\textbf{transformed $l_1$ functions} (TL1) by the difference of convex (DC) function algorithm  
provides a robust CS solution insensitive to the conditioning of $A$. 
Since verifying the incoherence condition like RIP or null space 
property \cite{CDD_09} on a specific matrix is NP 
hard, such robustness is a significant attribute of an algorithm}. 
Let us consider TL1 function of the form \cite{transformed-l1}: 
\be \label{TL1a}
	\rho_a(t)  =  \dfrac{(a+1)|t|}{a+|t|},   \  \ \;\;   \forall t \in \Re,
\ee 
with parameter $a \in (0,+\infty)$, see \cite{original TL1,DCA:le2013sparse} for alternative forms and 
the $l_0$ approximation property \cite{DCA:le2013sparse}. Another nice property of TL1 is that the TL1 proximal operator has closed form analytical 
solutions for all values of parameter $a$. Fast TL1 iterative theresholding algorithms have been 
devised and studied for both sparse vector and low rank matrix recovery problems 
lately \cite{TL1 threshold,TS1}. 
\medskip

The rest of the paper is organized as follows. In section 2, we study theoretical 
properties of TL1 penalty and TL1 regularized models in exact and stable recovery 
of $\ell_0$ minimal solutions. We show the advantage of TL1 over $\ell_1$ in exact 
recovery under linear constraint using the generalized null space 
property \cite{TW_17} and explicit examples. We analyze stable recovery for linear 
constraint with observation error based on a RIP condition.
We overcome the lack of scaling property of TL1 by introducing 
a normalized TL1 regularized problem. Though the RIP condition is 
not sharp, our analysis is the first of such kind for stable recovery by TL1.  
We also prove that the local minimizers of the TL1 
constrained model extract independent columns from the sensing matrix. 
In section 3, we present two DC algorithms for TL1 optimization (DCATL1). 
In section 4, we compare the performance of DCATL1 with 
some state-of-the-art methods using two classes of matrices: the Gaussian 
and the over-sampled discrete cosine transform (DCT) matrices. 
Numerical experiments indicate that DCATL1 is robust 
and consistently top ranked while maintaining high sparse 
recovery rates across sensing matrices of varying coherence. In section 5, we compare 
DCATL1 with DCA on other non-convex penalties such 
as PiE \cite{Nguyen-DCA-2015}, MCP \cite{MCP}, 
and SCAD \cite{fan2001variable}, and found DCATL1 to be competitive as well. 
Concluding remarks are in section 6. 

\section{Transformed $l_1$ (TL1) and its regularization models} 
\setcounter{equation}{0}

\begin{figure}
\def\arraystretch{3}
\begin{tabular}{l r}
\begin{minipage}[t]{0.45\linewidth}
\includegraphics[scale=0.35]{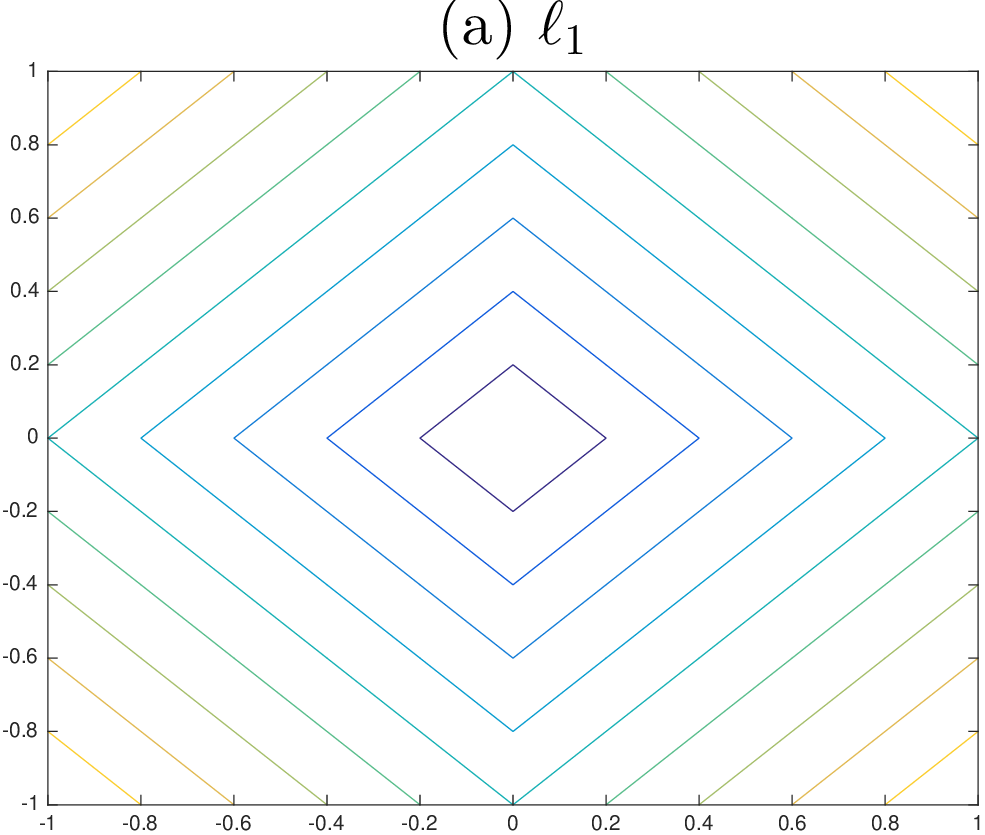} 
\end{minipage} & 
\begin{minipage}[t]{0.45\linewidth}
\includegraphics[scale= 0.35]{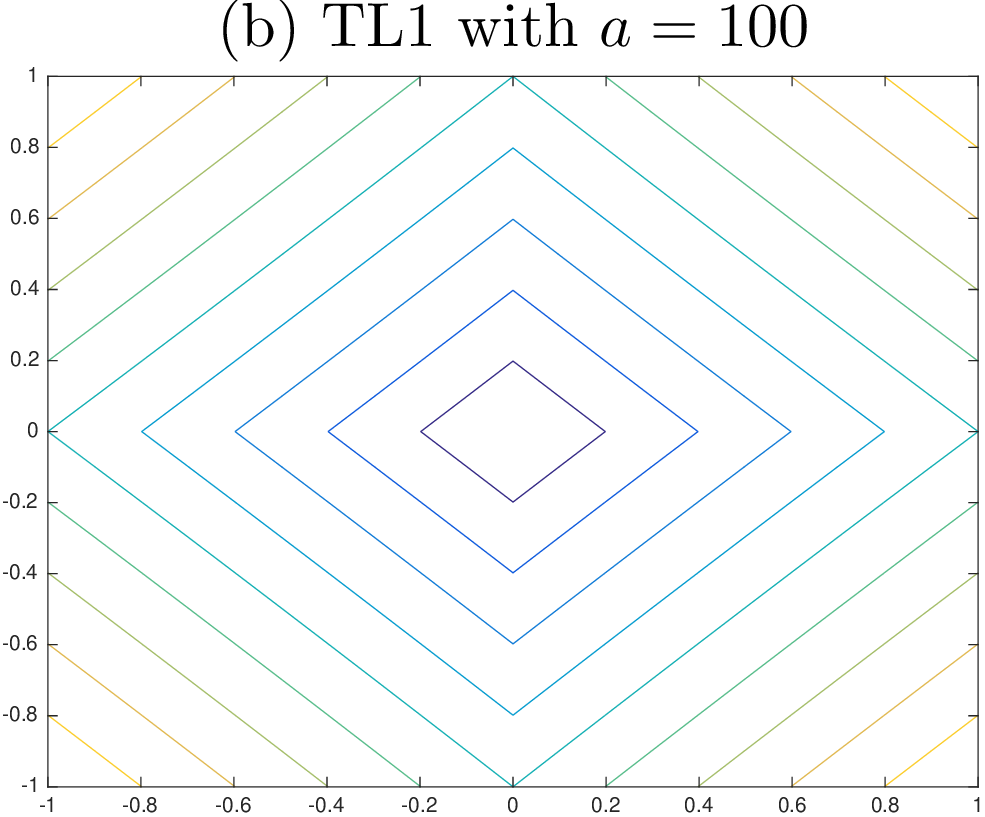} 
\end{minipage} \\
\begin{minipage}[t]{0.45\linewidth}
\includegraphics[scale=0.35]{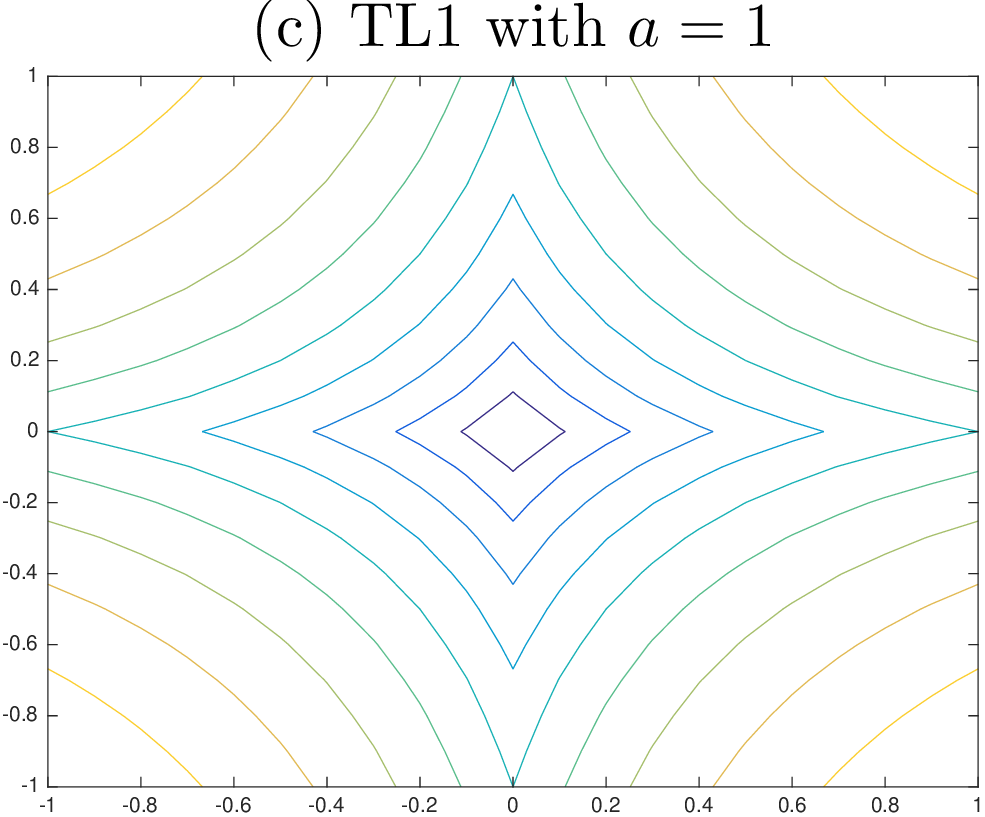} 
\end{minipage} & 
\begin{minipage}[t]{0.45\linewidth}
\includegraphics[scale= 0.35]{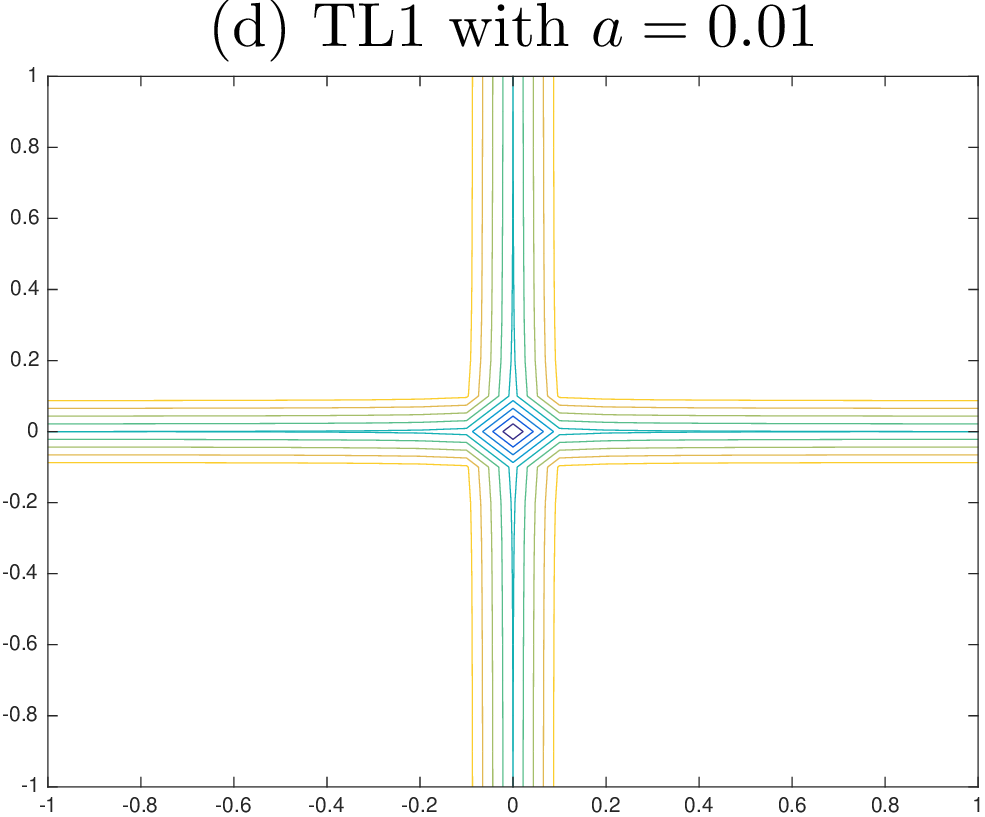} 
\end{minipage}
\end{tabular} 
\caption{Level lines of TL1 with different parameters: $a=100$ (figure b), $a=1$ (figure c), 
$a=0.01$ (figure d). For large parameter `$a$', the graph looks almost the same as $l_1$ (figure a). 
While for small value of `$a$', it tends to the axis.}
\label{gp: level lines}
\end{figure}

The TL1 penalty function $\rho_a(t)$ of (\ref{TL1a})
interpolates the $l_0$ and $l_1$ norms as  
\[
	\lim_{a \to 0^{+}} \rho_a(t) = \chi_{\{t \neq 0\}} 
	\;\; \rm{and} \;\;
	\lim_{a \to \infty} \rho_a(t) = |t|. 
\]

In Fig. (\ref{gp: level lines}), we compare level lines of $l_1$ and TL1 with 
different parameter values of `$a$'. 
With the adjustment of parameter `$a$', the TL1 can approximate 
both $l_1$ and $l_0$ well. 
Let us define TL1 regularization term $P_a(\cdot)$ as
\begin{equation}
P_a(x) = \sum \limits_{i = 1,...N} \rho_a(x_i).
\end{equation} 
In the following, we consider the constrained TL1 minimization model 
\begin{equation} \label{cons-optim}
   \min \limits_{x \in \Re^N} f(x) =  \min \limits_{x \in \Re^N} P_a(x)  \;\; s.t. \;\;\  Ax=y, 
\end{equation}
and the unconstrained TL1-regularized model
\begin{equation} \label{uncons-optim}
   \min \limits_{x \in \Re^N} f(x) = \min \limits_{x \in \Re^N} \frac{1}{2} \|Ax - y \|^2_2 + \lambda P_a(x).
\end{equation}

The following inequalities of $\rho_a$ will be used in the proof of TL1 theories.
\begin{lem}\label{lem1}

For $a\geq 0$, any $x_i$ and $x_j$ in $\Re$, the following inequalities hold: 
   \begin{equation}
   \rho_a(|x_i + x_j|) \leq \rho_a(|x_i| + |x_j|) \leq \rho_a(|x_i|) + \rho_a(|x_j|) 
   \leq 2\rho_a(\frac{|x_i| + |x_j|}{2}).
   \end{equation}
\end{lem}

\begin{proof}
   Let us prove these inequalities one by one, starting from the left.
\begin{itemize}
   \item[1)] 
      Note that $\rho_a(|t|)$ is increasing in the 
variable $|t|$. By triangle inequality $|x_i +x_j| \leq |x_i| + |x_j|$, we have: 
      \begin{equation*}
         \rho_a(|x_i + x_j|) \leq \rho_a(|x_i| + |x_j|).
      \end{equation*}        
   \item[2)]
      \begin{equation*}
      \begin{array}[c]{lll}
           \vspace{2mm}
         \rho_a(|x_i|) + \rho_a(|x_j|) 
           & = & \dfrac{(a+1)|x_i|}{a+|x_i|} + \dfrac{(a+1)|x_j|}{a+|x_j|} \\
           \vspace{2mm}
           & = & \dfrac{a(a+1)( |x_i|+|x_j|+2|x_i x_j|/a )}{a( a+|x_i|+|x_j|+|x_i x_j|/a )} \\
           \vspace{2mm}
           & \geq & \dfrac{(a+1)( |x_i|+|x_j|+|x_i x_j|/a )}{( a+|x_i|+|x_j|+|x_i x_j|/a )} \\
           \vspace{2mm}
           & = & \rho_a(|x_i|+|x_j|+|x_i x_j|/a) \\
           \vspace{2mm}
           & \geq & \rho_a(|x_i|+|x_j|). \\
      \end{array}
      \end{equation*}  
   \item[3)] 
      By concavity of the function $\rho_a(\cdot)$, 
      \begin{equation*}
         \dfrac{\rho_a(|x_i|) + \rho_a(|x_j|)}{2} \leq \rho_a \left (\dfrac{|x_i| + |x_j|}{2} \right ).
      \end{equation*}
\end{itemize}
\end{proof}

\begin{rmk} 
   \vspace{1 mm}
   It follows from Lemma \ref{lem1} that the triangular inequality holds for the 
    function $\rho(x) \equiv \rho_a(|x|)$:
   $\rho(x_i+x_j) = \rho_a(|x_i+x_j|) \leq \rho_a(|x_i|) + \rho_a(|x_j|) = \rho(x_i) + \rho(x_j)$. \\
   Also we have: $\rho(x) \geq 0$, and $\rho(x) = 0 \Leftrightarrow x=0 $. 
   Our penalty function $\rho$ acts almost like a norm. However, it lacks absolute scalability, 
   or $\rho(cx) \neq |c|\rho(x)$ in general. The next lemma  
   further analyzes this in terms of inequalities. 
\end{rmk}

\vspace{2mm}
\begin{lem} 
For scalar $t \in \Re$,
\begin{equation}
   \rho_a(|ct|) = \left\{ \begin{array}{ll}
         \leq |c|\rho_a(|t|)  & \mbox{if $|c| > 1$};\\
         \geq |c|\rho_a(|t|)  & \mbox{if $|c| \leq 1$}. 
         \end{array} \right. 
\end{equation}
   
\end{lem}

\begin{proof}    
   \begin{equation*}
   \begin{array}{lll}
      \rho_a(|ct|) & = & \dfrac{(a+1)\vert c \vert |t|}{a+\vert c \vert |t|} \\
                   & = & |c| \rho_a(|t|) \ \dfrac{a+|t|}{a+|c t|}. \\
   \end{array}
   \end{equation*} 
   So if $|c| \leq 1$, the factor $ \frac{a+|t|}{a+|ct|} \geq 1$. Then $\rho_a(|ct|) \geq |c| \rho_a(|t|)$. 
Similarly when $|c| > 1$, we have $\rho_a(|ct|) \leq |c| \rho_a(|t|)$.  
\end{proof}

\subsection{Exact and Stable Sparse Recovery for Constrained Model}
\medskip

For the constrained TL1 model (\ref{cons-optim}), 
we discuss the theoretical issue of exact and stable recovery 
of $l_0$ minimal solution. Specifically, 
let $x=\beta^0$ be the unique sparsest solution 
of $A x = y$ with $s$ nonzero components, 
we address whether it is possible to construct it 
by minimizing $P_a$.  
\medskip

Let $A$ be an $M\times N$ matrix, $T \subset \{ 1,...,N \}$, 
$A_T$ the matrix consisting of the columns $a_j$ of $A$, $j\in T$. 
Similarly for vector $x$, $x_T$ is a sub-vector consisting of components 
with indices in $T$.  Let vector $\beta$ be:
\begin{equation}
   \beta = {\rm arg} \min \limits_{x \in \Re^N} \{ P_a(x) | \ \ Ax=y \}.
\label{TL1_constrained}
\end{equation}
The necessary and sufficient condition of exact recovery, namely 
$\beta = \beta^0$, is the generalized null space property (gNSP, \cite{TW_17}):
\be
{\rm Ker}(A)\backslash \{0\} \subset {\rm gNS}:=
\{v \in \Re^N:\, P_a(v_T) < P_a(v_{T^c}),\, \forall T,\, |T|\leq s \}, 
\label{nsp1}
\ee
where $|T|$ is the cardinality (the number of elements) of the set $T$, 
and $T^c$ is the complement of $T$. 
The gNSP generalizes the well-known NSP for $\ell_1$ exact recovery \cite{CDD_09}:
\be
{\rm Ker}(A)\backslash \{0\} \subset {\rm NS}:=
\{v \in \Re^N:\, \|v_T\|_1 < \|v_{T^c}\|_1,\, \forall T,\, |T|\leq s \}.
\label{nsp2}
\ee
For the class of separable, concave and symmetric penalties \cite{TW_17} 
including $\ell_p$, TL1, capped $\ell_1$ ($\sum_i \, \min \, (|x_i|,\theta)$, $\theta >0$), 
SCAD, PiE, and MCP (see Table 5.1), \cite{TW_17} proved that gNSP is   
the necessary and sufficient condition for 
exact sparse recovery while being no more restrictive than NSP. 
In fact, the inclusion ${\rm NS}\subset {\rm gNS}$ holds for 
this class of penalties (Proposition 3.3, \cite{TW_17}). 
It follows that if exact recovery holds for a matrix $A$ by $\ell_1$, it 
is also true for any of these concave penalties. By a scaling argument, we show that:
\medskip

\begin{theo}
NSP of $\ell_1$ is equivalent to gNSP of SCAD or capped $\ell_1$. 
In other words, 
minimizing non-convex penalties SCAD and capped $\ell_1$ has no gain over minimizing $\ell_1$ 
in the exact sparse recovery problem.
\end{theo}
\medskip

\nit \begin{proof} Consider any $v \in {\rm Ker}(A)\backslash \{0\}$ satisfying gNSP (\ref{nsp1}). 
  Let $\epsilon \in (0,1)$ be less than $\theta/\|v\|_\infty $ ($\beta/\|v\|_\infty$) in case of capped-$\ell_1$ (SCAD). 
Then $\epsilon \, v\in {\rm Ker}(A)\backslash \{0\} $ also satisfies gNSP, and:
\[ P_a(\eps \, v_T) < P_a(\eps \, v_{T^c}),\;\; \forall T,\; |T|\leq s, \]
which is same as:
\[ \|\eps \, v_T\|_1 < \|\eps v_{T^c}\|_1,\;\; \forall T,\; |T|\leq s,\]
implying that $v$ satisfies NSP. Hence gNSP = NSP for SCAD and capped-$\ell_1$.   
\end{proof}
\medskip

We give an example of a square matrix below to show that the inclusion ${\rm NS}\subset {\rm gNS}$ is strict for  
 TL1, $\ell_p$, PiE, MCP, implying that the exact recovery by any one of these four penalties 
holds but that of $\ell_1$ (also SCAD, capped $\ell_1$) fails.
Consider:
\[ A=\left [ \begin{array}{lll}
 2 & 1 & 0 \\
 1 & 1 & 1 \\
0 & 1 & 2 
\end{array} \right] \]
with ${\rm Ker}(A)={\rm span}\{(1,-2,1)'\}$. The linear constraint is 
$A\, x = (1, 1, 1)'$. The sparsest solution is $\beta^0=(0, 1, 0)'$, $s=1$. 
The set $T$ has cardinality 1. If $T=\{2\}$, for any nonzero vector 
$ t \, (1, -2, 1)$ in ${\rm Ker}(A)$, $ t \not = 0$, 
$v_T= -2\, t $, $v_{T^c}=t\, (1, 1)'$. 
So $\|v_T\|_1 = 2\, |t|  = \|v_{T^c}\|_1$, NSP fails. 
To verify gNSP at $T=\{2\}$ for TL1, we have:
\[ P_a(v_T)=(a+1)2|t|/(a+2|t|)  < 2 (a+1)|t|/(a+|t|)=P_a(v_{T^c}). \] 
At $T=\{1\}$, $v_T=t$, $v_{T^c}= t [-2,1]$, we have:
\[ P_a(v_T)=(a+1)|t |/(a+|t|) < P_{a}(v_{T^c})= (a+1)2|t|/(a+2|t|)
+ (a+1)|t|/(a+|t|). \]
The case $ T=\{3\}$ is the same, and gNSP holds for TL1. 
\medskip

A similar verification on the validity of gNSP can be done 
for $\ell_p$, PiE and MCP. It suffices to check $T=\{2\}$ where the largest component of the null 
vector (in absolute value) is. For $\ell_p$ and PiE, $P_a(\cdot)$ is strictly concave on $\Re^{+}$, Jensen's inequality 
gives $P_a(v_T)= P_a(0) + P_a(2 |t|) < 2 P_a(|t|) = P_a(v_{T^c})$.
For MCP, $P_a$ is quadratic and strictly concave on $[0,\alpha \beta]$, hence 
$P_a(v_T) < P_a(v_{T^c})$ if $2 |t| \leq \alpha \beta $. If $|t|< \alpha \beta < 2|t|$, the line connecting 
$(0,0)$ and $(2|t|,P_a(2|t|))$ is still strictly below the $P_a$ curve, hence $P_a(v_T) <P_a(v_{T^c})$ holds. 
If $|t|\geq \alpha \beta$, $P_a(v_T)=P_a(2|t|) = \alpha \beta^2 < 2 \alpha \beta^2 = 2 P_a(|t|) = P_a(v_{T^c})$. 
\medskip

The example can be extended to a block diagonal 
$M\times M$ matrix ($M > 3$) of the form ${\rm diag}(A,B)$, where $B$ is any 
invertible $M-3$ square matrix, with the right hand side vector of the linear 
constraint being $(1, 1, 1, 0,\cdots, 0)'$.  The following is a 
rectangular $3\times 4$ matrix (in the class of fat sensing matrices of CS) 
where NSP fails and gNSP prevails: 
\[ A_f=\left [ \begin{array}{llll}
 2 & 1 & 0 & 1\\
 1 & 1 & 1 & 0 \\
0 & 1 & 2  & 0 
\end{array} \right] \]
for the linear constraint  
$A_f\, x = (1, 1, 1)'$, $x \in \Re^4$. 
The sparsest solution is $\beta^0=(0, 1, 0, 0)'$, $s=1$. 
The ${\rm Ker}(A_f) = {\rm span}\{(1,-2,1,0)'\}$, ${\rm rank}(A_f)=3$. 
Since the last component of any null vector is zero, NSP fails at $T=\{2\}$ 
as in the $3\times 3$ example while gNSP inequalities 
remain valid at $T \not = \{4\}$. Clearly, the gNSP inequality holds at $T=\{4\}$.  To summarize, we state:
\medskip

\begin{rmk}\label{gnsp}
The set of matrices satisfying gNSP of concave penalties ($\ell_p$, TL1, PiE, MCP) can be larger than that of NSP.  
Since matrices satisfying NSP or gNSP tend to be incoherent, we expect that the exact recovery by ($\ell_p$, TL1, PiE, MCP) 
is better than ($\ell_1$, capped $\ell_1$, SCAD) in this regime. This phenomenon is partly observed in our numerical experiments later. 
\end{rmk}
\medskip

Checking NSP or gNSP is NP hard in general.
The restricted isometry property (RIP) 
provides a sufficient condition for $\ell_1$ 
exact recovery or NSP, and is satisfied with overwhelming 
probability by Gaussian random matrices 
with i.i.d. entries \cite{candes2005error}. By the inclusion relation 
$NS \subset gNS$, the minimization on any one of the concave penalties above in the setting of (\ref{TL1_constrained}) 
recovers exactly the minimal $\ell_0$ solution $\beta^0$ for 
Gaussian random matrices with i.i.d. entries with overwhelming 
probability.   
\medskip

Though gNSP (\ref{nsp1}) is sharp for exact recovery, it is only applicable for 
precise measurement or when the linear constraint holds exactly. 
If there is any measurement error, can one recover the $\ell_0$ solution up 
to certain tolerance (error bound) ? To answer such stable recovery question for TL1,  
we carry out a RIP analysis below 
to show that a stable recovery of $\beta^0$ is possible 
based on a normalized TL1 minimization problem. Naturally, RIP 
analysis also gives an exact recovery result when the measurement is error free. 
Though sub-optimal, it is the first step towards a stable recovery theory.
To begin, we recall:
\medskip
 
\begin{defin} (Restricted Isometry Constant)
For each number $s$, define the $s$-restricted isometry constant of matrix $A$ 
as the smallest number $\delta_s \in (0,1)$ such that for all column index
subset $T$ with $|T| \leq s$ and all $x \in \Re^{|T|}$, the inequality
\[
      (1-\delta_s)\|x\|_2^2 \leq \| A_T x \|_2^2 \leq (1+\delta_s) \| x \|_2^2
\]
holds. The matrix $A$ is said to satisfy the s-RIP with $\delta_s$. 
\end{defin}

\vspace{2mm} 
Due to lack of scaling property of TL1, we introduce a normalization 
procedure to recover $\beta^0$. For a fixed $y$, the 
under-determined linear constraint has 
infinitely many solutions. Let $x^0$ be a solution 
of $Ax^0=y$, not necessarily the $l_0$ or $\rho_a$ minimizer. 
If $P_a(x^0) > 1$, we scale $y$ by a positive scalar $C$ as:
\begin{equation} \label{para: scaled pa}
   y_{C}  =  \frac{y}{C}; \ \ \  x_{C}  =  \frac{x^0}{C}.
\end{equation} 
Now $x_{C}$ is a solution to the 
equivalent scaled constraint: $Ax_{C} = y_{C}$. 
When $C$ becomes larger, the number $P_a(x_{C})$ 
is smaller and tends to $0$ in the limit $C \rightarrow \infty$. 
Thus, we can find a constant $C \geq 1$, such that $P_a(x_{C}) \leq 1$. 
That is to say, for scaled vector $x_C$, we always have: $ P_a(x_C) \leq 1$.  
Since the penalty $\rho_a(t)$ is increasing in positive variable $t$, we have: 
\[ P_a(x_{C})  \leq  |T^0| \rho_a( |x_{C}|_{\infty}  ) 
                     =  |T^0| \rho_a( \frac{|x^0|_{\infty}}{C}  ) 
                    = \dfrac{|T^0|(a+1)|x^0|_{\infty}}{aC + |x^0|_{\infty}} , 
\]
where $|T^0|$ is the cardinality of the support set $T^0$ of vector $x^0$. 
For $P_a ( x_{C} ) \leq 1$, 
\begin{equation*}
   \dfrac{|T^0|(a+1)|x^0|_{\infty}}{a\, C + |x^0|_{\infty}} \leq 1
\end{equation*}
suffices, or: 
\begin{equation} \label{Ccond}
   C \geq \dfrac{|x^0|_{\infty}}{a} \left(a\, |T^0| + |T^0| -1 \right). 
\end{equation}

Let $\beta^0$ be the $l_0$ minimizer for the constrained $l_0$ 
optimization problem (\ref{eq:l0 cons}) 
with support set $T$. Due to the scale-invariance of $l_0$, 
$\beta^0_{C}$ (defined similarly as above) is a global  
$l_0$ minimizer for the normalized problem: 
\begin{equation}\label{eq:l0 revised}
 \min \limits_{x} \ \ \| x \|_0, \;\; s.t. \;\;\ y_{C} = A x, 
\end{equation} 
with the same support set $T$. The exact recovery is 
stated below with proof in the appendix A 
for the normalized $\rho_a$ minimization problem:
\begin{equation}\label{eq:spar revised}
 \min \limits_{x} \ \ P_a(x), \;\; s.t. \;\;\ y_{C} = A x.
\end{equation}

\begin{theo}(Exact TL1 Sparse Recovery)\label{thm:global}

For a given sensing matrix $A$, let $\beta^0_{C}$ be a minimizer of 
(\ref{eq:l0 revised}), with $C$ satisfying (\ref{Ccond}). Let T be the support set
of $\beta^0_{C}$, with cardinality $|T|$. 
Suppose there is a number $R > |T|$ such that 
$b = (\frac{a}{a+1})^2 \frac{R}{|T|} > 1$ and 
\begin{equation}\label{RIP}
      \delta_{R} + b \; \delta_{R+|T|}\; < \; b - 1,
\end{equation} 
then the minimizer $\beta_{C}$ of (\ref{eq:spar revised}) is unique and equal 
to the minimizer $\beta^0_{C}$ in (\ref{eq:l0 revised}). Moreover, 
$C \, \beta_{C}$ is the unique minimizer of the $l_0$ 
minimization problem (\ref{eq:l0 cons}).
\end{theo}
\medskip

\vspace{2mm}
\begin{rmk} 

In Theorem \ref{thm:global}, if we choose $R = 3|T|$, the 
RIP condition (\ref{RIP}) is 
\[
	\delta_{3|T|} + 3\frac{a^2}{(a+1)^2} \delta_{4|T|} < 3\frac{a^2}{(a+1)^2} - 1.
\]
This inequality will approach $\delta_{3|T|} + 3\delta_{4|T|} < 2$ as parameter $a$ 
goes to $+\infty$, which is the RIP condition in \cite{candes2005error} satisfied 
by Gaussian random matrices 
with i.i.d. entries. The RIP condition (\ref{RIP}) is satisfied by the same class 
of Gaussian matrices when `a' is sufficiently large, 
though it is more stringent when `a' gets smaller. This is due to the  
lack of scaling property of the TL1 penalty and 
the sub-optimal treatment in the RIP analysis. Hence the true 
advantage of TL1 penalty for CS problems, to be seen in our numerical results later, 
is not reflected in the RIP condition. Theoretically, it is an open question 
to find random matrices that satisfy gNSP of TL1 but not NSP.   
\end{rmk}

\vspace{2mm}

The RIP analysis of exact TL1 recovery allows a stable recovery analysis 
stated below with proof in the appendix B. 
For a positive number $\tau$, 
we consider the problem: 
\begin{equation}\label{eq:spar revised noise}
 \min \limits_{x} \ \ P_a(x), \;\; s.t. \;\;\ \| y_{C} - Ax \|_2 \leq \tau.
\end{equation} 

\begin{theo}(Stable TL1 Sparse Recovery)
Under the same RIP condition (\ref{RIP}) in Theorem \ref{thm:global}, 
the solution $\beta^n_{C}$ of the problem (\ref{eq:spar revised noise}) satisfies the 
inequality
\begin{equation*}
   \| \beta^n_{C} - \beta^0_{C}\|_2 \leq D \tau,  
\end{equation*}
for a positive constant $D$ depending only on $\delta_R$ and $\delta_{R+|T|}$.
\end{theo}
\medskip

\subsection{Sparsity of Local Minimizer}
\medskip

We study properties of local minimizers of both the  
constrained problem (\ref{cons-optim}) and the unconstrained 
model (\ref{uncons-optim}). As in $l_p$ and $l_{1-2}$ minimization 
\cite{l1-l2-yinminimization,l1-l2-lou2014computing}, a local minimizer of 
TL1 minimization extracts linearly independent columns from the 
sensing matrix $A$, without requiring $A$ to satisfy NSP.  
\medskip

\begin{theo} (Local minimizer of constrained model)

Suppose $x^*$ is a local minimizer of the constrained problem (\ref{cons-optim}) and $T^* = supp(x^*)$, then $A_{T^*}$ is of full column rank, i.e. columns of $A_{T^*}$ are linearly independent. 
\end{theo}

\begin{proof}
Here we argue by contradiction. Suppose that the column vectors of $A_{T^*}$ are 
not linearly independent, then there exists non-zero vector $v \in ker(A)$, such 
that $supp(v) \subseteq T^*$. For any neighbourhood of $x^*$, $N(x^*,r)$, we can scale $v$ so that:
\begin{equation} \label{ineq:local v ball}
   \|v\|_2 \leq \min \{ r; \ \ |x_{i}^{*}|, \; i \in T^*\}.
\end{equation}
  
Next we define: 
\[
\xi_1 = x^* + v, \;\; \xi_2 = x^* - v,
\] 
so $\xi_1, \xi_2 \in \mathscr{B}(x^*,r)$,  
and $x^* = \frac{1}{2} (\xi_1 + \xi_2)$. On the other hand, 
from $supp(v) \subseteq T^*$, we have that $supp(\xi_1), supp(\xi_2) \subseteq T^*$. 
 Moreover, due to the inequality (\ref{ineq:local v ball}), vectors $x^*$, $\xi_1$, and 
$\xi_2$ are located in the same orthant, i.e. $sign(x^*_i) = sign(\xi_{1,i}) = sign(\xi_{2,i})$, for 
any index $i$. It means that $\frac{1}{2}|\xi_1| + \frac{1}{2}|\xi_2| = \frac{1}{2}|\xi_1 + \xi_2|$. 
Since the penalty function $P_a(t)$ is strictly concave for non-negative variable $t$, 
\begin{equation*}
\begin{array}{lll}
\frac{1}{2}P_a(\xi_1) + \frac{1}{2}P_a(\xi_2) & = & \frac{1}{2}P_a(|\xi_1|) + \frac{1}{2}P_a(|\xi_2|) \\
& < & P_a( \frac{1}{2}|\xi_1| + \frac{1}{2}|\xi_2| ) = P_a(\frac{1}{2}|\xi_1 + \xi_2|) = P_a(x^*).  
\end{array}
\end{equation*} 
So for any fixed $r$, we can find two vectors $\xi_1$ and $\xi_2$ in 
the neighbourhood $\mathscr{B}(x^*,r)$, such that 
$\min \{ P_a(\xi_1), P_a(\xi_2) \} \leq \frac{1}{2}P_a(\xi_1) + \frac{1}{2}P_a(\xi_2) < P_a(x^*)$. 
Both vectors are in the feasible set of the 
constrained problem (\ref{cons-optim}), in contradiction with  
the assumption that $x^*$ is a local minimizer.  
\end{proof}
\medskip

The same property also holds for the local minimizers of unconstrained model (\ref{uncons-optim}),
because a local minimizer of the unconstrained problem is also a local minimizer for 
a constrained optimization model \cite{candes2005error, l1-l2-yinminimization}. 
We skip the details and state the result below. 
\medskip
  
\begin{theo} (Local minimizer of unconstrained model)

Suppose $x^*$ is a local minimizer of the unconstrained problem (\ref{uncons-optim}) and $T^* = supp(x^*)$, then columns of $A_{T^*}$ are linearly independent.  
\end{theo}
\medskip

\begin{rmk}  \label{cor: local}

From the two theorems above, we conclude the following: 
\begin{itemize}
\item[(i)] 
	For any local minimizer of (\ref{cons-optim}) or (\ref{uncons-optim}), 
	e.g. $x^*$, the sparsity of $x^*$ is at most rank(A);
\item[(ii)] 
	The number of local minimizers is finite, for both problem (\ref{cons-optim}) and (\ref{uncons-optim}). 
\end{itemize}
\end{rmk}

\section{DC Algorithm for Transformed $l_1$ Penalty}
\setcounter{equation}{0}
\medskip

DC (Difference of Convex functions) programming and DCA (DC Algorithms) were introduced 
in 1985 by Pham Dinh Tao, and extensively developed by Le Thi Hoai An, 
Pham Dinh Tao and their coworkers to become a useful tool for non-convex optimization 
and sparse signal recovery (\cite{DCA-ong2013learning,DCA:le2013sparse,DCA:le2014,DCA:le2015} and references therein).
A standard DC program is of the form 
\begin{equation*}
   \alpha = \inf \{ f(x) = g(x) - h(x): x \in \Re^N \}   \ \ \ \ \ \ \ \ \ \ (P_{dc}),
\end{equation*}
where $g$, $h$ are lower semicontinuous proper convex functions on $\Re^n$. Here $f$ is called a 
DC function, while $g-h$ is a DC decomposition of $f$.
\medskip

The DCA is an iterative method and generates a sequence $\{ x^k \}$. 
At the current point $x^l$ of iteration, function $h(x)$ is approximated
by its affine minorization $h_l(x)$, defined by 
\begin{equation*}
   h_l(x) = h(x^l) + \langle x-x^l,y^l \rangle, \ \ \ y^l \in \partial h(x^l), 
\end{equation*}
where the subdifferential $\partial h(x)$ at 
$x \in dom \, (h)$ is the closed convex set: 
\begin{equation}
   \partial h(x) := \{ y \in \Re^{N} : h(z) \geq h(x) + \langle z-x, y \rangle, 
\ \ \forall z\in \Re^N \},
\end{equation}
which generalizes the derivative in the sense that 
$h$ is differentiable at $x$ if and only if $\partial h(x)$ is 
a singleton or $\{ \nabla h(x) \}$.  
The minorization gives a convex program of the form: 
\begin{equation*}
   \inf \{ g(x) - h_l(x): x \in \Re^N \} \Leftrightarrow  
   \inf \{ g(x) - \langle x,y^l \rangle: x \in \Re^N \},
\end{equation*}
where the optimal solution is denoted as $x^{l+1}$.   
\medskip

In the following, we present DCAs for TL1 regularized problems, see related 
DCAs in \cite{DCA:le2013sparse,DCA:le2014}. We refer to \cite{DCA:le2015} 
for DCAs on general sparse penalty regularized problems and the 
consistency analysis (convergence of global minimizers of the regularized problems 
to the $l_0$ minimizers).
\medskip

\subsection{DC Form of TL1}
\medskip

The TL1 penalty function $p_a(\cdot)$ is written as a 
difference of two convex functions: 
\begin{equation}
\begin{array}[c]{lll}
\vspace{1mm}
 \rho_a(t) & = & \dfrac{(a+1)|t|}{a+|t|}  \\ 
 \vspace{1mm}
           & = & \dfrac{(a+1)|t|}{a} 
           		- \left(\ \dfrac{(a+1)|t|}{a} - \dfrac{(a+1)|t|}{a+|t|} \ \right)  \\
           \vspace{1mm}
           & = & \dfrac{(a+1)|t|}{a} - \dfrac{(a+1)t^2}{a(a+|t|)},  
\end{array}
\end{equation}
where the second term is $C^1$. 
The general derivative of function $P_a(\cdot)$ is: 
\begin{equation} \label{equ: d P_a}
	\partial P_a(x) = \frac{a+1}{a} \partial\, \|x\|_1 
		- \nabla \, \varphi_a(x),
\end{equation}
where:
\be
\varphi_a (x) = \sum_{i=1}^{N} \frac{(a+1)|x_i|^2}{a(a+|x_i|)}   \label{phia}
\ee
is a $C^1$ function with regular gradient, 
and $\partial \|x\|_1$ is the 
subdifferential of $\|x\|_1$, i.e. 
$\partial \|x\|_1 = \{ sgn(x_i) \}_{i=1,..., N}$, where
\begin{equation}\label{equ: sgn}
   sgn(t) = \begin{cases}
              sign(t) , & \mathrm{if} \ \ t \neq 0,  \\
              [-1,1]  , & \mathrm{otherwise}. \\
            \end{cases}
\end{equation}

\subsection{Algorithm for Unconstrained Model --- DCATL1}
For the unconstrained optimization problem (\ref{uncons-optim}): 
\begin{equation*}
   \min \limits_{x \in \Re^N} f(x) =  \min \limits_{x \in \Re^N} \frac{1}{2} \|Ax - y \|^2_2 + \lambda P_a(x),
\end{equation*}
a DC decomposition is 
$f(x)  =  g(x) - h(x)$,
where
\begin{equation} \label{equ:f=g-h}
\left\{
\begin{array}{lll}   
   g(x) & = & \dfrac{1}{2} \|Ax-y\|^2_2 + c\|x\|_2^2 + \lambda \dfrac{(a+1)}{a} \|x\|_1; \\
   h(x) & = & \lambda \varphi_a(x) + c\|x\|_2^2. \\
\end{array}
\right.
\end{equation}
Here the function $\varphi_a(x)$ is defined in equation (\ref{phia}). 
Additional factor $c\|x\|_2^2$ with positive hyperparameter $c > 0$ is used to improve 
the convexity of these two functions, and will be used in the convergence theorem.



\begin{algorithm}[H]
\caption{DCA for unconstrained transformed $l1$ penalty minimization} 
\label{alg: DCA} 
\begin{algorithmic} 
\STATE Define: \ \ $\epsilon_{outer} > 0$ 
\STATE Initialize: \ \ $x^0 = 0, n = 0$ 
\WHILE{$ |x^{n+1} - x^n| > \epsilon_{outer} $} 
  \STATE $ v^n \in \partial h(x^n) = \lambda \, \nabla \varphi_a(x^n) + 2c x^n $ 
  \STATE $ x^{n+1} =  arg \min \limits_{x \in \Re^N } \{ \frac{1}{2} \|Ax-y\|^2_2 
                      + c\|x\|^2 + \lambda \dfrac{(a+1)}{a} \|x\|_1 
                      - \langle x,v^n \rangle \}$             
  \STATE then $n+1 \rightarrow n$
\ENDWHILE 
\end{algorithmic} 
\end{algorithm} 

At each step, we solve a strongly convex $l_1$-regularized sub-problem: 
\begin{equation} \label{ADMM: sub}
\begin{array}{lll}
   x^{n+1} & = & arg \min \limits_{x \in \Re^N } \{ \frac{1}{2} \|Ax-y\|^2_2 
                      + c\|x\|^2 + \lambda \dfrac{(a+1)}{a} \|x\|_1 
                      - \langle x,v^n \rangle \}  \\
           & = & arg \min \limits_{x \in \Re^N } \{ \frac{1}{2} x^t(A^tA+2cI)x 
                      - \langle x,v^n + A^t y  \rangle \ 
                      + \lambda \dfrac{(a+1)}{a} \|x\|_1 \}.
\end{array}
\end{equation}
We now employ the Alternating Direction Method of Multipliers (ADMM), \cite{Boyd_11}. 
After introducing a new variable $z$, the sub-problem is recast as: 
\begin{equation}
\begin{array}{r}
  \min \limits_{x, z \in \Re^N } \{\ \ \frac{1}{2} x^t(A^tA+2cI)x 
                      - \langle x,v^n + A^t y  \rangle \ 
                      + \lambda \dfrac{(a+1)}{a} \|z\|_1  \}  \\
   s.t. \  x-z = 0.
\end{array}
\end{equation}

Define the augmented Lagrangian function as:
\begin{equation*}
   L(x,z,u) = \frac{1}{2} x^t(A^tA+2cI)x 
                      - \langle x,v^n + A^t y  \rangle \ 
                      + \lambda \dfrac{(a+1)}{a} \|z\|_1 + \frac{\delta}{2} \| x - z \|_2^2 + u^t(x-z),
\end{equation*}
where $u$ is the Lagrange multiplier, and $\delta > 0$ is a penalty parameter.
The ADMM consists of three iterations: 
\begin{equation*} 
\left \{
\begin{array}{rrl} 
  x^{k+1} & = & arg \min \limits_x \ \ L(x,z^k,u^k); \\
  z^{k+1} & = & arg \min \limits_z \ \ L(x^{k+1},z,u^k); \\
  u^{k+1} & = & u^k + \delta (x^{k+1} - z^{k+1}).
\end{array}
\right.
\end{equation*} 
The first two steps have closed-form solutions and are described 
in Algorithm \ref{alg:ADMM}, where $shrink(\cdot,\cdot)$ is a soft-thresholding operator given by: 
\begin{equation*}
   shrink(x,r)_i = sgn(x_i) \max\{ |x_i| -r, 0 \}. 
\end{equation*}

\begin{algorithm}[H]
\caption{ADMM for subproblem (\ref{ADMM: sub})}  
\label{alg:ADMM}
\begin{algorithmic} 
\STATE Initial guess: $x^0$, $z^0$, $u^0$ and iterative index $k = 0$
\WHILE{ not converged} 
  \STATE  $ x^{k+1} := (A^tA + 2cI + \delta I)^{-1} (A^t y - v^n + \delta z^k- u^k)$
  \STATE  $ z^{k+1} := shrink( \ \ x^{k+1} + u^k, \frac{a+1}{a\delta}\lambda \ \ )$
  \STATE  $ u^{k+1} := u^k + \delta (x^{k+1} - z^{k+1}) $
  \STATE then $k+1 \rightarrow k$
\ENDWHILE 
\end{algorithmic} 
\end{algorithm} 
\medskip

\subsection{Convergence Theory for Unconstrained DCATL1}
\medskip

We present a convergence theory for the Algorithm \ref{alg: DCA} (DCATL1). 
We prove that the sequence $\{ f(x^n) \}$ is decreasing and convergent, 
while the sequence $\{ x^n \}$ is bounded under some requirement on $\lambda$. 
Its subsequential limit vector $x^*$ is a stationary point  
satisfying the first order optimality condition. 
Our proof is based on the convergence theory of $l_1 - l_2$ penalty function 
\cite{l1-l2-yinminimization} besides the general DCA results \cite{DCA:tao1997convex, DCA:tao1998dc}.    
\medskip

\begin{defin} (Modulus of strong convexity)
For a convex function $f(x)$ , the modulus of strong convexity of $f$ 
on $\Re^N$, denoted as $m(f)$, is defined by 
\begin{equation*}
   m(f) := \sup\{ \rho > 0 : f - \frac{\rho}{2} \|\cdot\|_2^2 \text{ is convex on } \Re^N \}.
\end{equation*}
\end{defin}
\medskip

Let us recall an inequality from Proposition A.1 in \cite{DCA:tao1998dc} concerning the 
sequence $f(x^n)$. 
\begin{lem} \label{lem:f(x)}
Suppose that $f(x) = g(x) - h(x)$ is a D.C. decomposition, and the sequence $\{ x^n \}$ is generated 
by (\ref{ADMM: sub}), then 
\begin{equation*}
   f(x^n) - f(x^{n+1}) \geq \dfrac{m(g) + m(h)}{2} \| x^{n+1} - x^n \|_2^2. 
\end{equation*} 
\end{lem}
\medskip

The convergence theory is below for our unconstrained Algorithm \ref{alg: DCA} --- DCATL1.  
The objective function is : $f(x) = \frac{1}{2} \|Ax - y \|^2_2 + \lambda P_a(x)$. 
\medskip

\begin{theo}
The sequences $\{ x^n \}$ and $\{ f(x^n) \}$ 
in Algorithm \ref{alg: DCA} satisfy: 
\begin{enumerate}
\item Sequence $\{ f(x^n) \}$ is decreasing and convergent. 
\item $ \| x^{n+1} - x^n \|_2 \rightarrow 0 $ as $n \rightarrow \infty$. If 
$\lambda > \dfrac{\|y\|_2^2}{2(a+1)} $, $\{ x^n \}_{n=1}^{\infty}$ is bounded. 
\vspace{1mm}
\item Any subsequential limit vector $x^*$ of $\{ x^n \}$ satisfies the first 
order optimality condition: 
\begin{equation} \label{equ:opti-cond}
	0 \in A^T(Ax^* - y) + \lambda \partial P_a(x^*),
\end{equation}
implying that $x^*$ is a stationary point of (\ref{uncons-optim}).
\end{enumerate}

\end{theo}

\vspace{3mm}

\begin{proof}
\begin{enumerate}
\item 
By the definition of $g(x)$ and $h(x)$ in equation (\ref{equ:f=g-h}),  it is easy to see that: 
\begin{equation*}
\begin{array}[c]{lll}
   m(g) & \geq & 2c; \\
   m(h) & \geq & 2c.
\end{array}
\end{equation*}
By Lemma \ref{lem:f(x)}, we have: 
\begin{equation*}
\begin{array}[c]{lll}
   f(x^n) - f(x^{n+1}) & \geq & \dfrac{m(g) + m(h)}{2} \| x^{n+1} - x^n \|_2^2 \\
                       & \geq & 2c \| x^{n+1} - x^n \|_2^2.
\end{array}
\end{equation*}
So the sequence $\{ f(x^n) \}$ is decreasing and non-negative, thus convergent. 

\item 
\medskip
It follows from the convergence of $\{ f(x^n) \}$ that: 
\begin{equation*}
   \| x^{n+1} - x^n \|_2^2 \leq \frac{f(x^{n}) - f(x^{n+1})}{2c} \rightarrow 0, \ \ as \ n \rightarrow \infty. 
\end{equation*}

If $y=0$, since the initial vector $x^0 = 0$, and the sequence $\{ f(x^n) \}$ is decreasing, 
we have $f(x^n) = 0$, $\forall n \geq 1$. So $x^n = 0$, and the boundedness holds. 

Consider non-zero vector $y$. Then 
\begin{equation*}
	f(x^n) = \frac{1}{2} \|Ax^n - y \|^2_2 + \lambda P_a(x^n) 
		\leq f(x^0) = \frac{1}{2} \|y\|_2^2.
\end{equation*}
So $ \lambda P_a(x^n) \leq \frac{1}{2} \|y\|_2^2 $, implying 
$2 \lambda \rho_a( \|x^n \|_{\infty} ) \leq  \|y\|_2^2$, or:
\[  
	{ 2\lambda (a+1) || x^n \|_{\infty} \over a + \| x^n \|_{\infty} } 
		\leq \|y \|_{2}^{2}. 
\]
If $ \lambda > \dfrac{\|y\|_2^2}{2(a+1)}$, then
\begin{equation*}
	\|x^n\|_{\infty} \leq \dfrac{a \|y\|_2^2 }{2 \lambda (a+1) - \|y\|_2^2}.
\end{equation*}
Thus the sequence $\{ x^n \}_{n=1}^{\infty}$ is bounded.

\item 
Let $\{ x^{n_k} \}$ be a subsequence of $\{x^n\}$ which converges to $x^*$. 
So the optimality condition at the $n_k$-th step of Algorithm \ref{alg: DCA} is expressed as: 
\begin{equation} \label{equ:nk-opti}
\begin{array}{ll}
   0 \in & A^T(Ax^{n_k} - y) + 2c(x^{n_k} - x^{n_k-1}) \\
   		& + \lambda (\frac{a+1}{a}) \partial \|x^{n_k}\|_1 - \lambda \nabla \varphi_a(x^{n_k-1}).
\end{array}
\end{equation}
Since $ \| x^{n+1} - x^n \|_2 \rightarrow 0 $ as $n \rightarrow \infty$ and 
$x^{n_k}$ converges to $x^*$, as shown in 
Proposition 3.1 of \cite{l1-l2-yinminimization}, we have that for 
sufficiently large index $n_k$, 
\begin{equation*}
   \partial \| x^{n_k} \|_1 \subseteq \partial \| x^* \|_1.
\end{equation*}
Letting $n_k \rightarrow \infty$ in (\ref{equ:nk-opti}), we have 
\[
	0 \in A^T(Ax^* - y) + \lambda (\frac{a+1}{a}) \partial \|x^*\|_1 
		- \lambda \nabla \varphi_a(x^*).
\]
By the definition of $\partial P_a(x)$ at (\ref{equ: d P_a}), we have
$
	 0 \in A^T(Ax^* - y) + \lambda \partial P_a(x^*) .
$ 

\end{enumerate}
\end{proof}

\subsection{Algorithm for Constrained Model}

Here we also give a DCA scheme to solve the constrained problem (\ref{cons-optim})
\begin{equation*}
\begin{array}{l}
   \min \limits_{x \in \Re^N} P_a(x) \;\; s.t. \;\;\  Ax=y. \\
   \Leftrightarrow \\
   \min \limits_{x \in \Re^N} \dfrac{a+1}{a}\| x \|_1  - \varphi_a(x)  
   \;\; s.t. \;\;\ Ax=y.
\end{array}
\end{equation*}
We can rewrite the above optimization as
\begin{equation} \label{scheme: cons DCA}
	\min \limits_{x \in \Re^N}
	 \dfrac{a+1}{a}\| x \|_1 + \chi(x)_{\{Ax=y\}} - \varphi_a(x)
	 = g(x) - h(x),
\end{equation}
where $g(x) = \dfrac{a+1}{a}\| x \|_1 + \chi(x)_{\{Ax=y\}}$ is a polyhedral convex function
\cite{DCA:tao1997convex}. 

Let $z = \nabla \varphi_a(x)$, 
then the convex sub-problem is: 
\begin{equation} \label{equ: cons dual}
   \min \limits_{x \in \Re^N} \frac{a+1}{a}\| x \|_1 - \langle z, x \rangle   
   \;\; s.t. \;\;\ Ax=y.
\end{equation}
To solve (\ref{equ: cons dual}), we introduce two Lagrange multipliers $u,v$ and define an augmented Lagrangian: 
\begin{equation*}
   L_{\delta}(x,w,u,v) = \frac{a+1}{a} \|w\|_1 - z^tx + u^t(x-w) + v^t(Ax-y) + \frac{\delta}{2} \|x-w\|^2
   + \frac{\delta}{2}\|Ax-y\|^2,
\end{equation*}
where $\delta > 0$. ADMM finds a saddle point $(x^*,w^*,u^*,v^*)$, such that: 
\begin{equation*}
   L_{\delta}(x^*,w^*,u,v) \leq L_{\delta}(x^*,w^*,u^*,v^*) \leq L_{\delta}(x,w,u^*,v^*)  \ \ \  \forall x,w,u,v
\end{equation*}
by alternately minimizing $L_{\delta}$ with respect to $x$, 
minimizing with respect to $y$ and updating the dual variables $u$ and $v$. 
The saddle point $x^*$ will be a solution to (\ref{equ: cons dual}). 
The overall algorithm for solving the constrained TL1 is described in Algorithm (\ref{alg:cons ADMM}).

\begin{algorithm}[H]
\caption{DCA method for constrained TL1 minimization}  
\label{alg:cons ADMM}
\begin{algorithmic} 
\STATE Define $\epsilon_{outer} > 0$, $\epsilon_{inner} > 0$.
Initialize $x^0 = 0$ and outer loop index $n = 0$
\WHILE{ $\|x^n - x^{n+1}\| \geq \epsilon_{outer}$}
\STATE $z = \nabla \varphi_a(x^n)$  \\
   Initialization of inner loop: 
   $x^0_{in} = w^0 = x^n$, $v^0 =0$ and $u^0=0$. \\
   Set inner index $j=0$.
  \WHILE {$\|x^j_{in} - x^{j+1}\| \geq \epsilon_{inner}$}  
  \STATE  $ x^{j+1}_{in} := (A^t A + I)^{-1} ( w^j + A^ty + \frac{z - u^j - A^t v^j}{\delta})$ 
  \STATE  $ w^j = shrink( \ \ x^{j+1}_{in} + \frac{u^j}{\delta}, \frac{a+1}{a\delta} \ \ )$ 
  \STATE  $ u^{j+1} := u^j + \delta(x^{j+1} - w^j ) $
  \STATE  $ v^{j+1} := v^j + \delta(Ax^{j+1} - y) $
  \ENDWHILE  
\STATE $x^n = x^j_{in} $ and $n = n+1$.
\ENDWHILE
\end{algorithmic} 
\end{algorithm} 
\medskip

According to DC decomposition scheme (\ref{scheme: cons DCA}), Algorithm 3 
is a polyhedral DC program. Similar convergence theorem as the unconstrained model 
in the last section can be proved. Furthermore, due to property of polyhedral DC programs,
this constrained DCA also has a finite convergence. It means that if the inner 
subproblem (\ref{equ: cons dual}) is exactly solved, $\{ x^n \}$, the sequence generated by this iterative DC 
algorithm, has finite subsequential limit points \cite{DCA:tao1997convex}.

\section{Numerical Results}

\setcounter{equation}{0}

In this section, we use two classes of randomly generated matrices 
to illustrate the effectiveness of our Algorithms: 
DCATL1 (difference convex algorithm for transformed $l_1$ penalty) 
and its constrained version. We compare them separately with 
several state-of-the-art solvers on recovering sparse vectors:
\medskip

\begin{itemize}
\item unconstrained algorithms: 
	\begin{enumerate}
	\item[(i)] Reweighted $l_{1/2}$ \cite{reweighted-l1/2-WTYing-2013improved};
	\item[(ii)] DCA $l_{1-2}$ algorithm \cite{l1-l2-yinminimization, l1-l2-lou2014computing};	
	\item[(iii)] CEL0 \cite{CEL0}
	\end{enumerate}
\item constrained algorithms: 
	\begin{enumerate}
	\item[(i)] Bregman algorithm \cite{Bregman:yin2008};
	\item[(ii)] Yall1; 
	\item[(iii)] $Lp-RLS$ \cite{lprls-chartrand2008iteratively}.	
	\end{enumerate}
\end{itemize} 
All our tests were performed on a $Lenovo$ desktop with 16 GB of RAM and Intel Core processor  
$i7-4770$ with CPU at $3.40GHz \times 8 $ under 64-bit Ubuntu system. \\

The two classes of random matrices are: 
\begin{itemize}
   \item[1)] Gaussian matrix.
   \item[2)] Over-sampled DCT with factor $F$.
\end{itemize}
\medskip

We did not use prior information of the true sparsity of the original signal $x^*$. 
Also, for all the tests, the computation is initialized with zero vectors. 
In fact, the DCATL1 does not guarantee a global minimum in general, due to 
nonconvexity of the problem. Indeed we observe that DCATL1 with 
random starts often gets stuck at local minima especially when the matrix $A$ 
is ill-conditioned (e.g. $A$ has a large condition number or is highly coherent). 
In the numerical experiments, by setting $x_0 = 0$, we find that 
DCATL1 usually produces an optimal solution, exactly or almost equal to the ground truth vector. 
The intuition behind our choice is that by using zero vector as initial guess, the first step
 of our algorithm reduces to solving an unconstrained weighted $l_1$ problem. 
 So basically we are minimizing TL1 on the basis of $l_1$,
 which is why minimization of TL1 initialized by $x_0 = 0$ always outperforms 
$l_1$, see \cite{YX_17} for a rigorous analysis.


\subsection{Choice of Parameter: `$a$'}

In DCATL1, parameter $a$ is also very important. When $a$ tends to zero, 
the penalty function approaches the $l_0$ norm. If $a$ goes to $+\infty$, 
objective function will be more convex and act like the $l_1$ optimization. 
So choosing a better $a$ will improve the effectiveness and success rate for our algorithm. 

\begin{figure}
\centering
\includegraphics[scale=0.7]{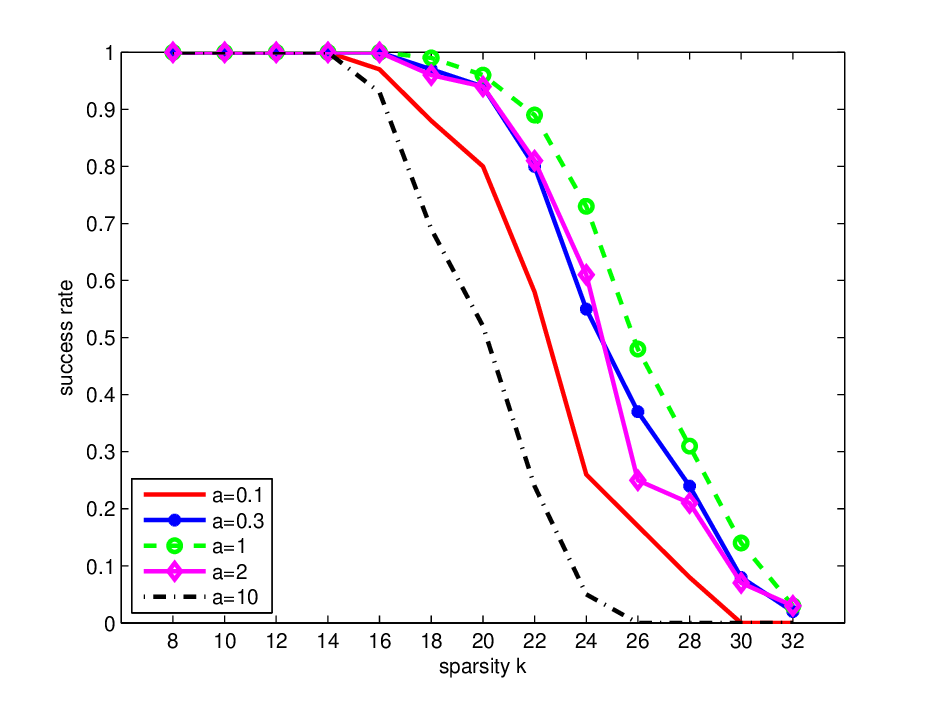}
\caption{Numerical tests on parameter $a$ with $M = 64$, $N = 256$ by the 
unconstrained DCATL1 method.}
\label{fig:choice of a}
\end{figure}
\medskip

We tested DCATL1 on recovering sparse vectors with different parameter $a$, 
varying among $\{0.1\ \ 0.3\ \ 1\ \ 2\ \ 10 \}$. In this test, $A$ is a 
$64 \times 256$ random matrix generated by normal Gaussian distribution. 
The true vector $x^*$ is also a randomly generated sparse vector with 
sparsity $k$ in the set $\{ 8 \ \ 10 \ \ 12 \ . . . \ \ 32 \}$. Here the regularization
parameter $\lambda$ was set to be $10^{-5}$ for all tests. Although the best $\lambda$ 
may be $k$-dependent in general, we considered the noiseless case and 
chose $\lambda = 10^{-5}$ (small and fixed) to approximately enforce $Ax = Ax^*$. 
For each $a$, we sampled 100 times with different $A$ and $x^*$. 
The recovered vector $x_r$ is accepted and recorded as one success if the 
relative error: $ \frac{ \|x_r - x^*\|_2 }{\|x^*\|_2} \leq 10^{-3} $.
\medskip

Fig. \ref{fig:choice of a} shows the success rate using DCATL1 
over 100 independent trials for various values of parameter $a$ and sparsity $k$. 
From the figure, we see that DCATL1 with $a=1$ is the best 
among all tested values. Also numerical results for $a = 0.3$ and $a=2$ (near 1), 
are better than those with 0.1 and 10. This is because the objective function 
is more non-convex at a smaller $a$ and thus more difficult to solve. On the 
other hand, iterations are more likely to stop at a local $\ell_1$ minima far from $\ell_0$ solution 
if $a$ is too large. Thus in all the following tests, we set the parameter $a=1$. 

\subsection{Numerical Experiment for Unconstrained Algorithm}
The stopping conditions for outer loops are relative iteration error 
$\frac{\| x^{n+1} - x^{n}\|_2}{\| x^{n+1}\|_2} < 10^{-5}$ and maximum 
iteration steps $20$. While, for the inner loop, the stopping condition are 
relative iteration error $10^{-8}$ and maximum iteration steps $5000$.
The methods in comparison methods are applied with default parameters. 

\subsubsection{Gaussian matrix}
\medskip

We use $\mathcal{N}(0,\Sigma)$, the multi-variable normal distribution 
to generate Gaussian matrix $A$. Here covariance matrix is  
$\Sigma = \{ (1-r)*\chi_{(i=j)} + r  \}_{i,j} $, where 
the value of `$r$' varies from 0 to 0.8. 
In theory, the larger the $r$ is, the more difficult it is 
to recover true sparse vector. For matrix $A$, the row number 
and column number are set to be $M = 64$ and $N = 1024$. 
The sparsity $k$ varies among $\{ 5 \ \ 7 \ \ 9 \ . . . \ \ 25 \}$. 
\medskip

We compare four algorithms in terms of success rate. 
Denote $x_r$ as a reconstructed solution by a certain algorithm. 
We consider one algorithm to be successful, if the relative error of $x_r$ to 
the truth solution $x$ is less that 0.001, $i.e.$, $\frac{\| x_r -x \|_2}{\| x \|_2} < 10^{-3}$. 
In order to improve success rates for all compared algorithms, we set tolerance parameter 
to be smaller or maximum cycle number to be higher inside each algorithm. 
\medskip

The success rate of each algorithm is plotted in 
Figure \ref{figure:Gaussian} with parameter $r$ from  
the set: $\{ 0 \ \ 0.2 \ \ 0.6 \ \ 0.8 \}$. 
For all cases, DCATL1 and reweighted $l_{1/2}$ algorithms (IRucLq-v) performed almost 
the same and both were much better than the other two, 
while the CEL0 has the lowest success rate.  

\begin{figure}
\begin{tabular}{lr}
\begin{minipage}[t]{0.5\linewidth}
\includegraphics[scale=0.35]{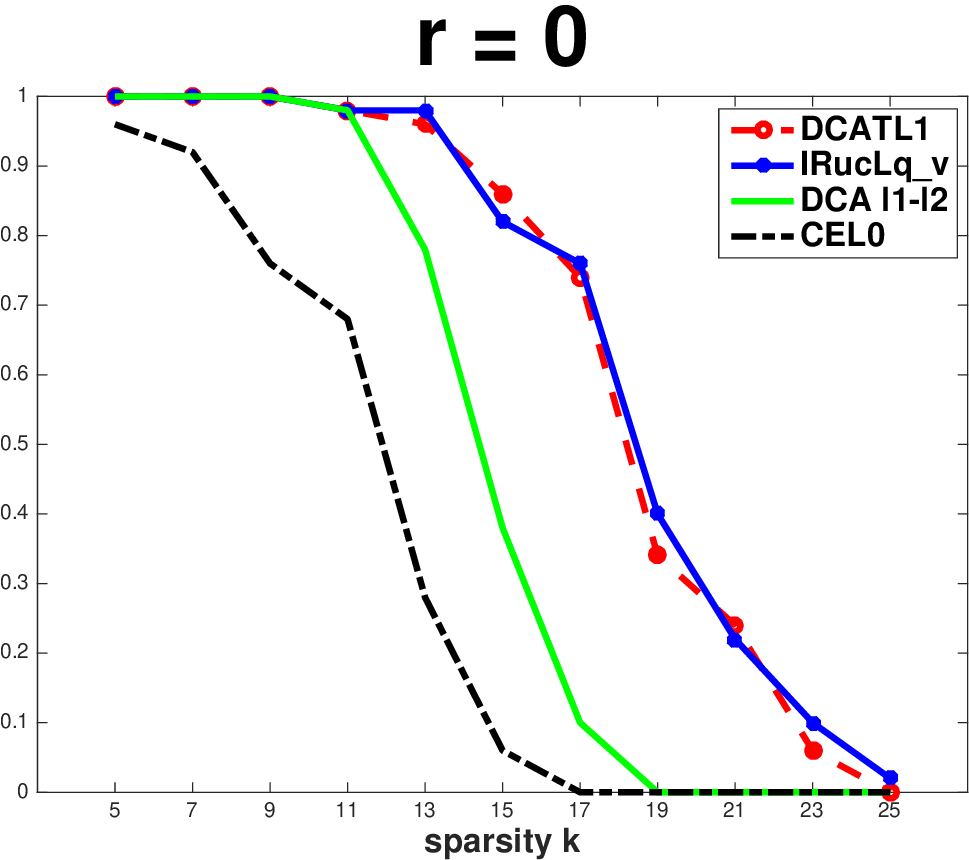}
\end{minipage}  &
\begin{minipage}[t]{0.5\linewidth}
\includegraphics[scale= 0.35]{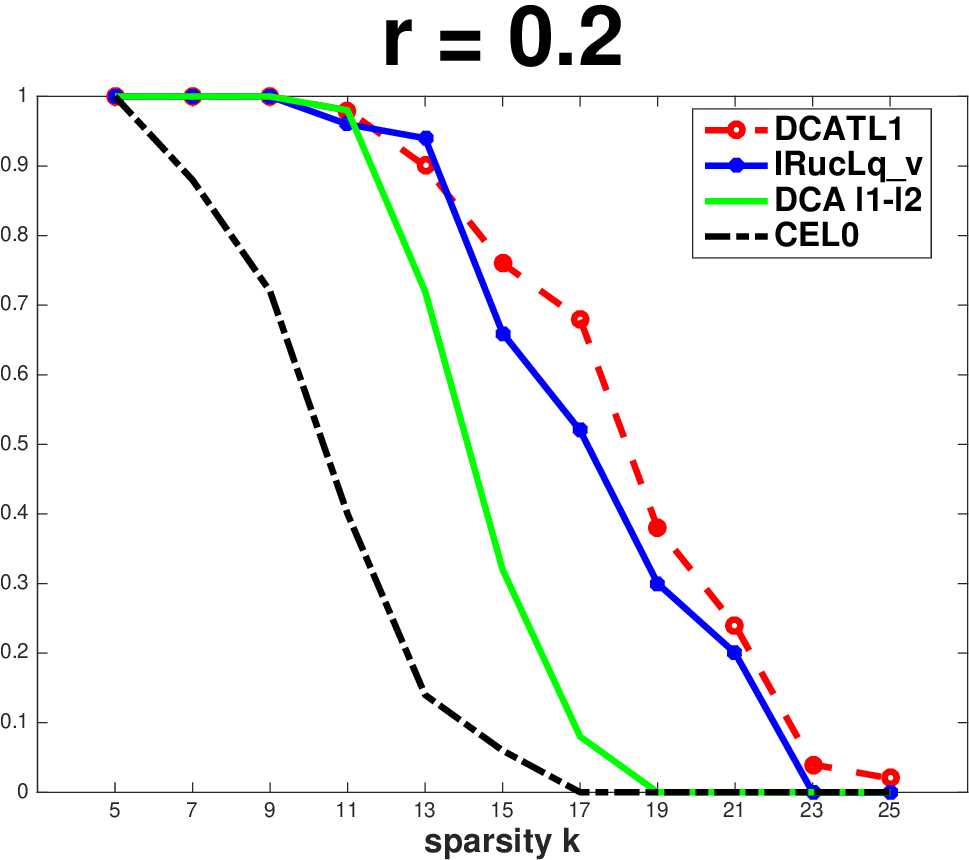}
\end{minipage} \\
\begin{minipage}[t]{0.5\linewidth}
\includegraphics[scale=0.35]{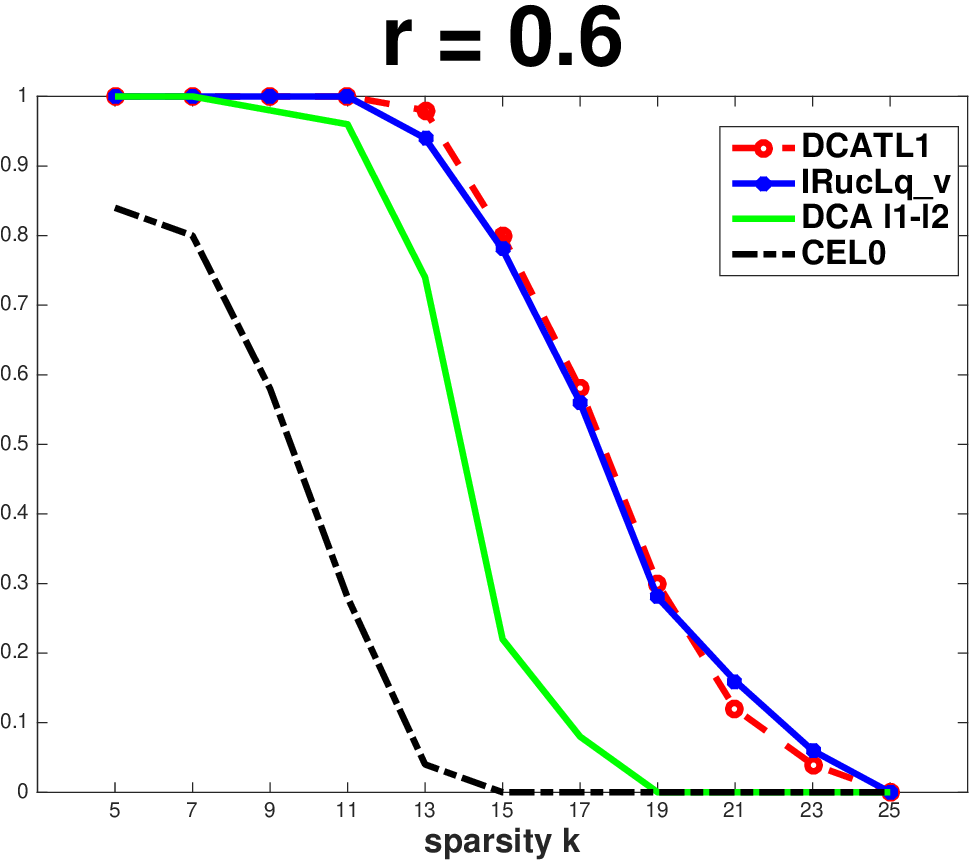}
\end{minipage}  &
\begin{minipage}[t]{0.5\linewidth}
\includegraphics[scale= 0.35]{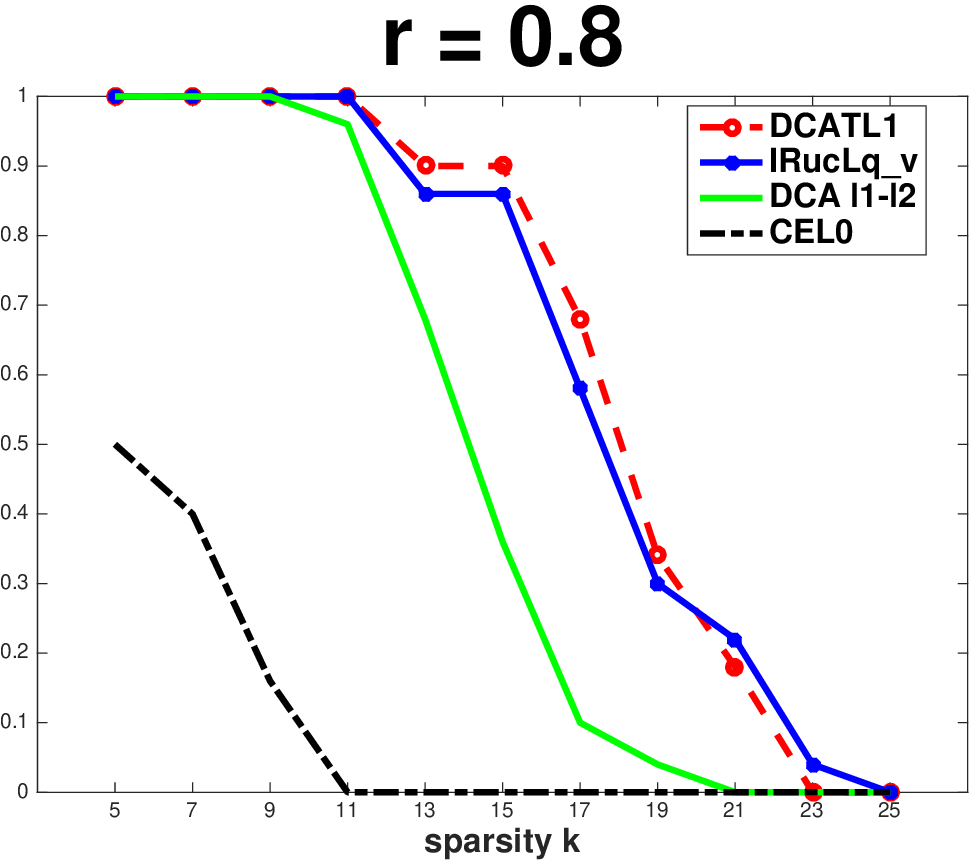}
\end{minipage}  \\
\end{tabular}
\caption{Numerical tests for unconstrained algorithms under Gaussian 
generated matrices: $M = 64$, $N = 1024$ with different coherence $r$.}
\label{figure:Gaussian}
\end{figure}

\subsubsection{Over-sampled DCT}
\medskip
 
The over-sampled DCT matrices \cite{Fann_12,l1-l2-lou2014computing,l1-l2-yinminimization} are: 
\begin{equation}
\begin{array}[c]{l}
   A  =  [a_1,...,a_N] \in \Re^{M \times N}, \\
    \ \ a_j = \dfrac{1}{\sqrt{M}} \cos(\dfrac{2 \pi \omega (j-1)}{F}), \ \ j = 1,...,N,   \\
   \text{ and $\omega$ is a random vector, drawn uniformly from $(0,1)^M$}.     
\end{array}
\end{equation}

Such matrices appear as the real part of the 
complex discrete Fourier matrices in spectral estimation \cite{Fann_12}.  
An important property  
is their high coherence: for a $100 \times 1000$ matrix with $F =10$, 
the coherence is 0.9981, while the coherence of the same size matrix with $F = 20$, is typically 0.9999.
\medskip

The sparse recovery under such matrices is possible only if the non-zero elements 
of solution $x$ are sufficiently separated. This phenomenon is characterized 
as $minimum$ $separation$ in \cite{Super:candes2013mini}, and this minimum 
length is referred as the Rayleigh length (RL). The value of RL for 
matrix $A$ is equal to the factor $F$. It is closely related to the coherence 
in the sense that larger $F$ corresponds to larger coherence of a matrix. 
We find empirically that at least 2RL is necessary to ensure optimal sparse recovery 
with spikes further apart for more coherent matrices.
\medskip

Under the assumption of sparse signal with 2RL separated spikes, we compare 
those four algorithms in terms of success rate. Denote $x_r$ as a reconstructed 
solution by a certain algorithm. We consider one algorithm successful, 
if the relative error of $x_r$ to the truth solution $x$ is less that $10^{-3}$, 
$i.e.$, $\frac{\| x_r -x \|_2}{\| x \|_2} < 10^{-3}$. The success rate is averaged 
over 50 random realizations.
\medskip

Fig. \ref{figure:F} shows success rates for those algorithms with increasing 
factor $F$ from 2 to 20. The sensing matrix is of size $100 \times 1500$. 
It is interesting to see that along with the increasing of value $F$, 
DCA of $l_{1}-l_{2}$ algorithm performs better and better, especially 
after $F \geq 10$, and it has the highest success rate among all. 
Meanwhile, reweighted $l_{1/2}$ is better for low coherent matrices. 
When $F \geq 10$, it is almost impossible for it to recover sparse solution 
for the highly coherent matrix. Our DCATL1, however, is more robust and 
consistently performed near the top, sometimes even the best. 
So it is a valuable choice for solving sparse optimization 
problems where coherence of sensing matrix is unknown. 
\medskip 

We further look at the success rates of DCATL1 with different combinations 
of sparsity and separation lengths for the over-sampled DCT matrix $A$. The rates are recorded 
in Table \ref{table:RL}, which shows that when the separation is above with the 
minimum length, the sparsity relative to $M$ 
plays more important role in determining the success rates of recovery.

\begin{table}
\caption{The success rates (\%) of DCATL1 for different combination of sparsity and minimum separation lengths.} 
\label{table:RL}
\centering 
\begin{tabular}{c rrrrrr} 
\hline\hline 
sparsity & 5 & 8 & 11 & 14 & 17 & 20 \\ [0.5ex] 
\hline 
1RL & 100 &  100  &  95  &  70  &  22  &   0 \\
2RL & 100 &  100  &  98  &  74  &  19  &   5 \\
3RL & 100 &  100  &  97  &  71  &  19  &   3 \\
4RL & 100 &  100  & 100  &  71  &  20  &   1 \\
5RL & 100 &  100  &  96  &  70  &  28  &   1 \\
[1ex] 
\hline 
\end{tabular} 
\label{tab:hresult} 
\end{table}

\begin{figure}
\begin{tabular}{lll}
\begin{minipage}[t]{0.32\linewidth}
\includegraphics[scale=0.26]{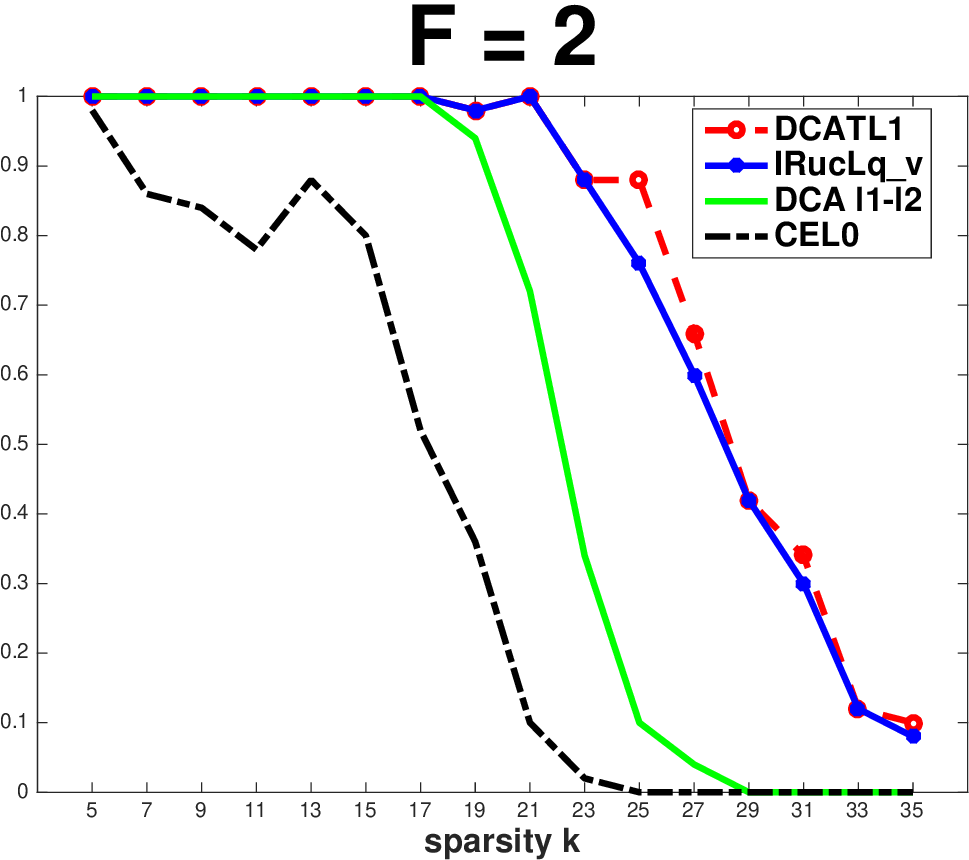}
\end{minipage} &
\begin{minipage}[t]{0.32\linewidth}
\includegraphics[scale= 0.26]{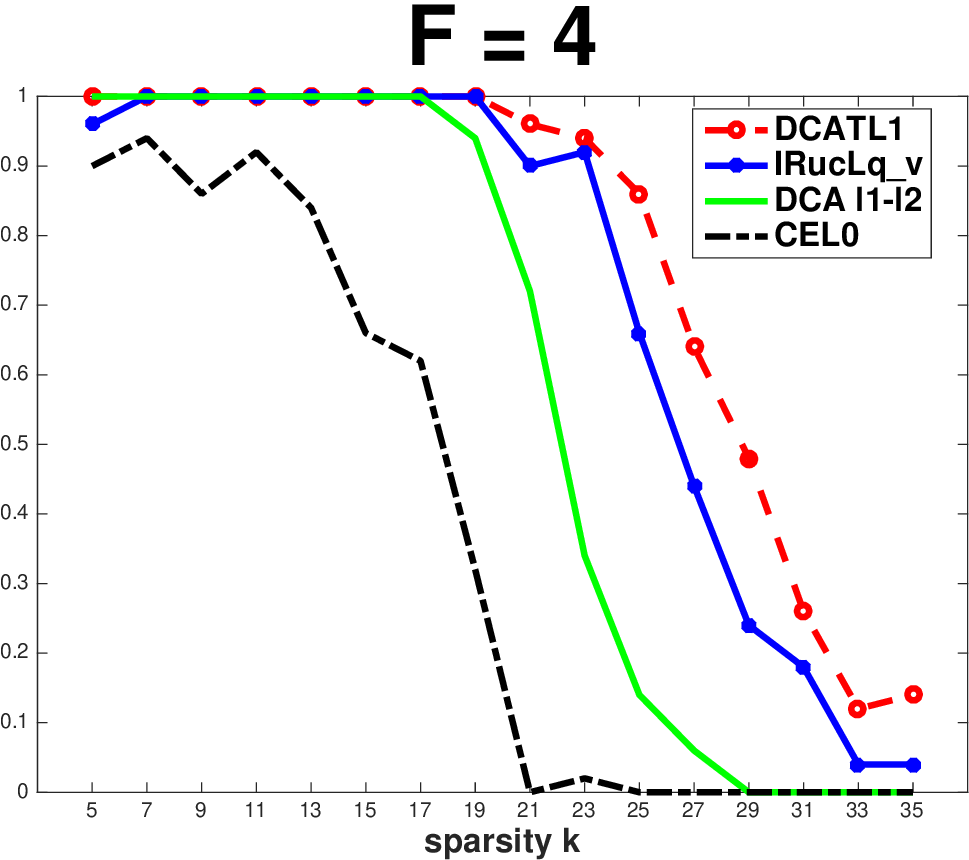}
\end{minipage} &
\begin{minipage}[t]{0.32\linewidth}
\includegraphics[scale=0.26]{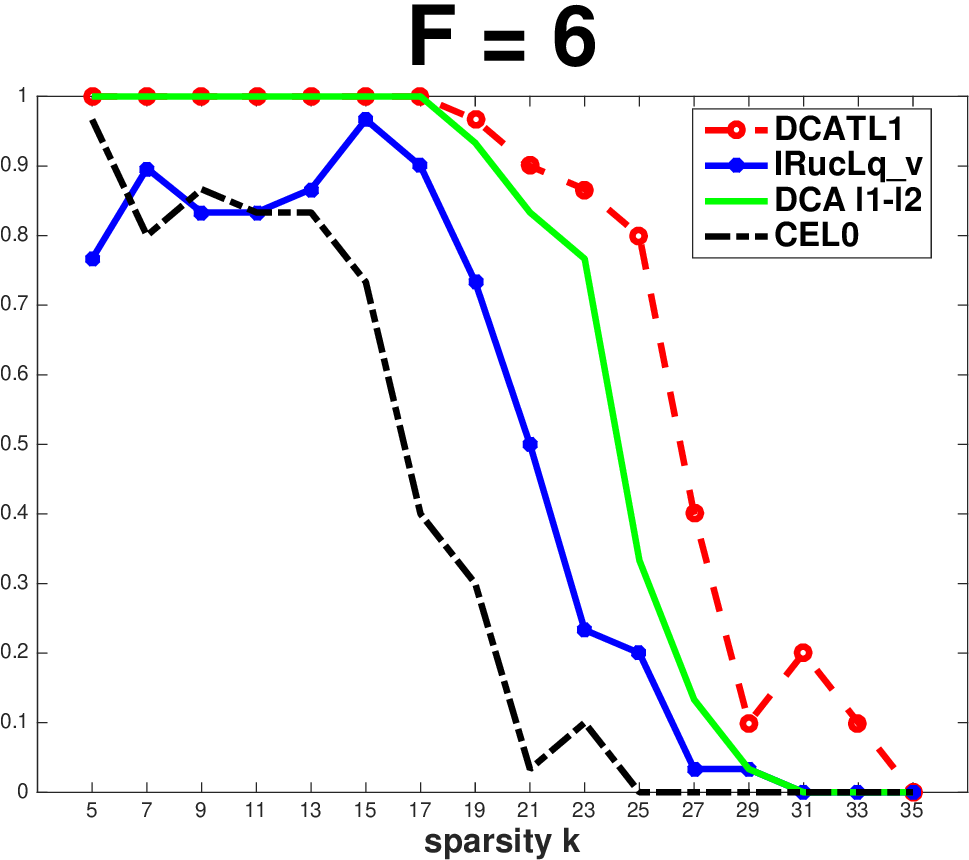}
\end{minipage} \\
\begin{minipage}[t]{0.32\linewidth}
\includegraphics[scale= 0.26]{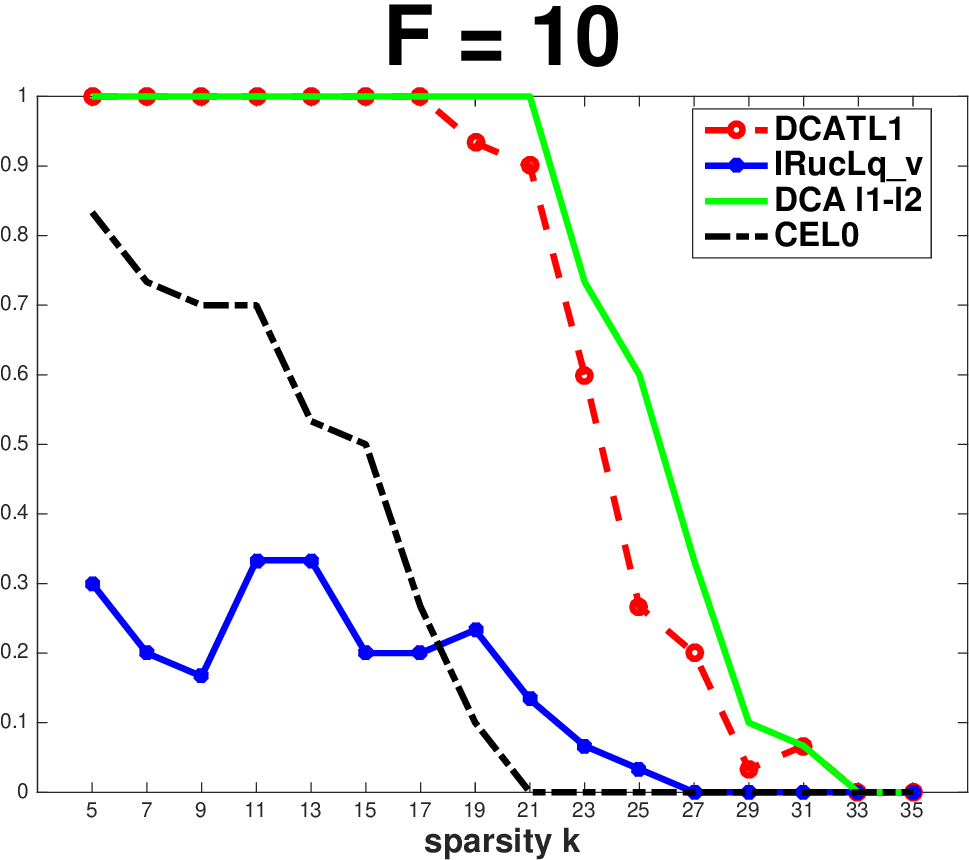}
\end{minipage}  &
\begin{minipage}[t]{0.32\linewidth}
\includegraphics[scale=0.26]{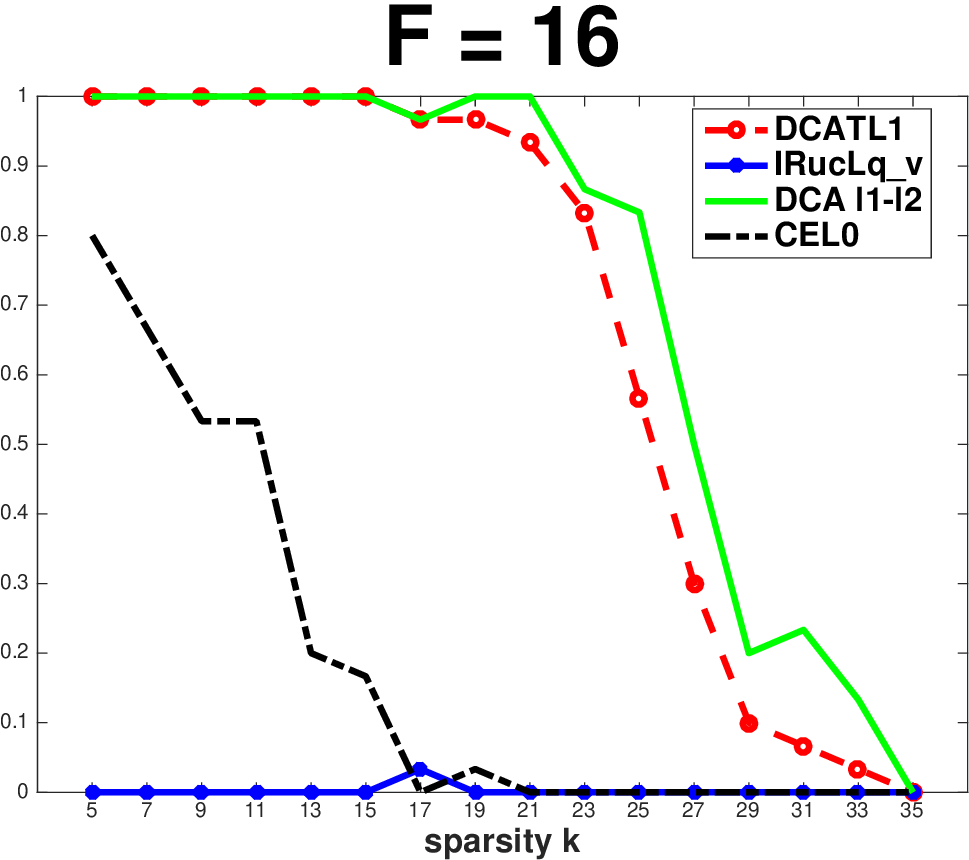}
\end{minipage} &
\begin{minipage}[t]{0.32\linewidth}
\includegraphics[scale= 0.26]{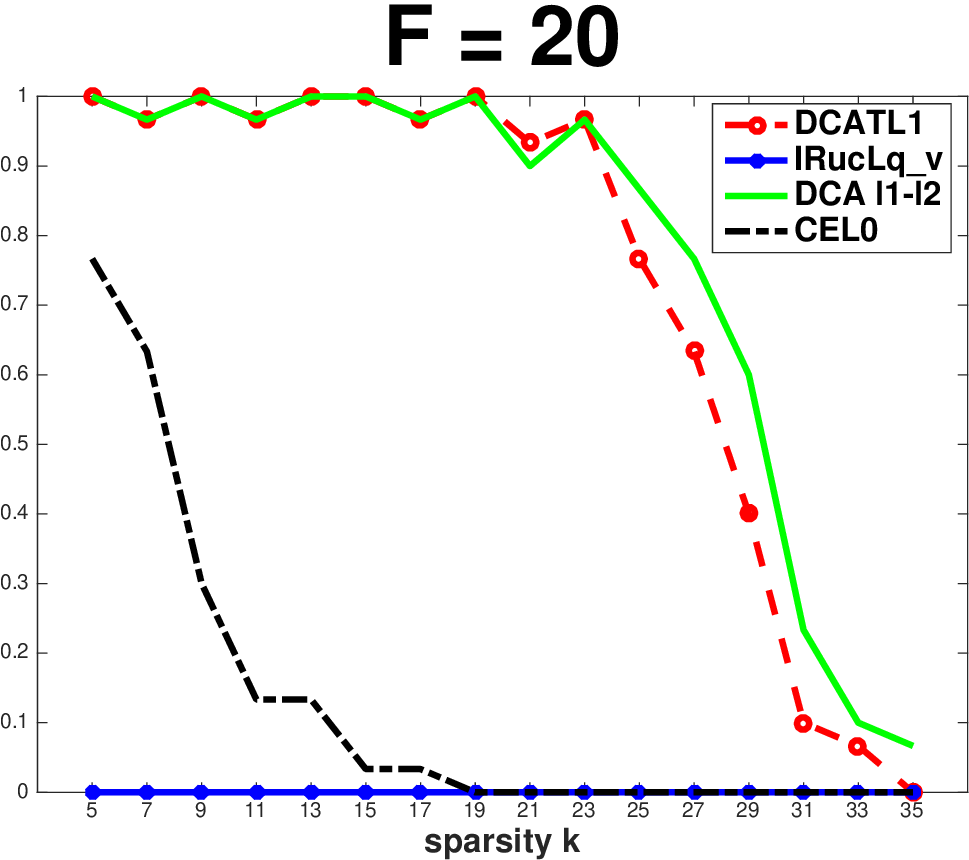}
\end{minipage} 
\end{tabular}
\caption{Numerical test for unconstrained algorithms under over-sampled DCT matrices: $M = 100$, $N = 1500$ 
with different $F$, and peaks of solutions separated by $2RL = 2F$.}
\label{figure:F}
\end{figure}

\subsection{Numerical Experiment for Constrained Algorithm}
For constrained algorithms, we performed similar numerical experiments. 
An algorithm is considered successful 
if the relative error of the numerical result $x_r$ from the ground truth $x$ 
is less than $10^{-3}$, or $\frac{\| x_r -x \|_2}{\| x \|_2} < 10^{-3}$. 
We did 50 trials to compute average success rates for all the numerical experiments as 
for the unconstrained algorithms.  

The stopping conditions for outer loop are relative iteration error 
$\frac{\| x^{n+1} - x^{n}\|_2}{\| x^{n+1}\|_2} < 10^{-5}$ and maximum 
iteration steps $20$. While, for the inner loop, the stopping condition are 
relative iteration error $10^{-5}$ and maximum iteration steps $1000$.
For other comparison methods, they are applied with default parameters. 

\subsubsection{Gaussian Random Matrices}
We fix parameters $(M,N)=(64,1024)$, while covariance parameter $r$ is varied from 0 to 0.8. 
Comparison is with the reweighted $l_{1/2}$ and two $l_1$ algorithms (Bregman and yall1). 
In Fig. (\ref{graph:cons-gaussian}), we see that $Lp-RLS$ is 
the best among the four algorthms with DCATL1 trailing not much behind. 

\begin{figure}
\begin{tabular}{lr}
\begin{minipage}[t]{0.5\linewidth}
\includegraphics[scale=0.35]{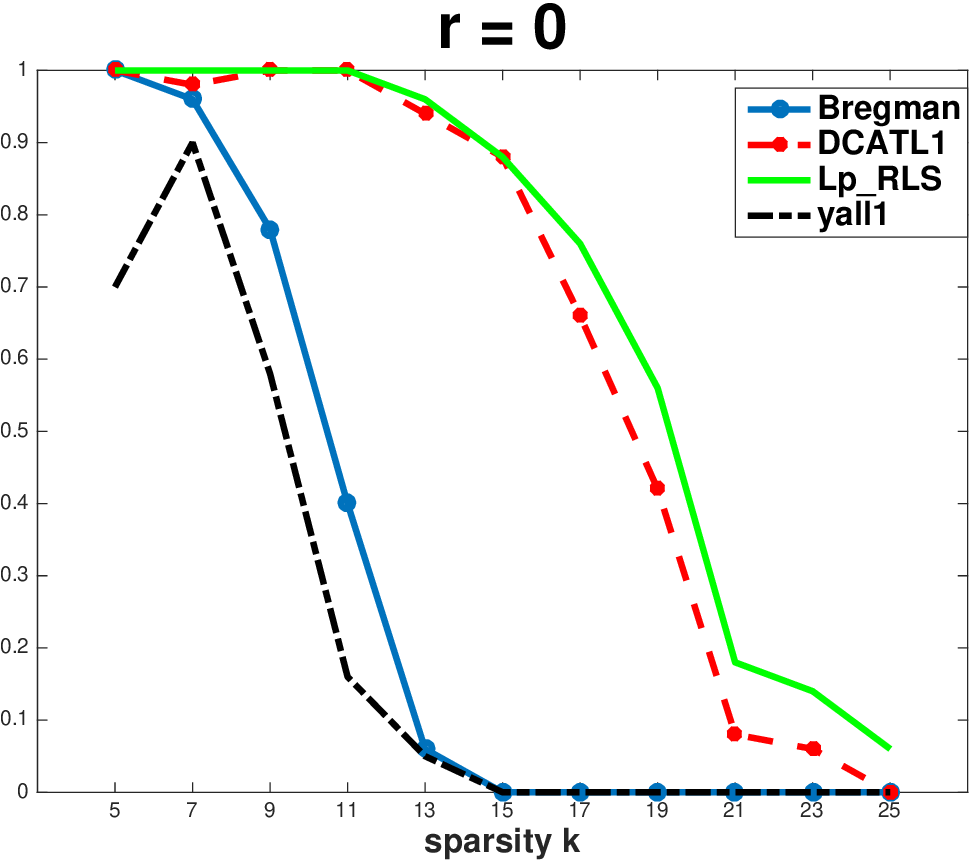}
\end{minipage}  &
\begin{minipage}[t]{0.5\linewidth}
\includegraphics[scale= 0.35]{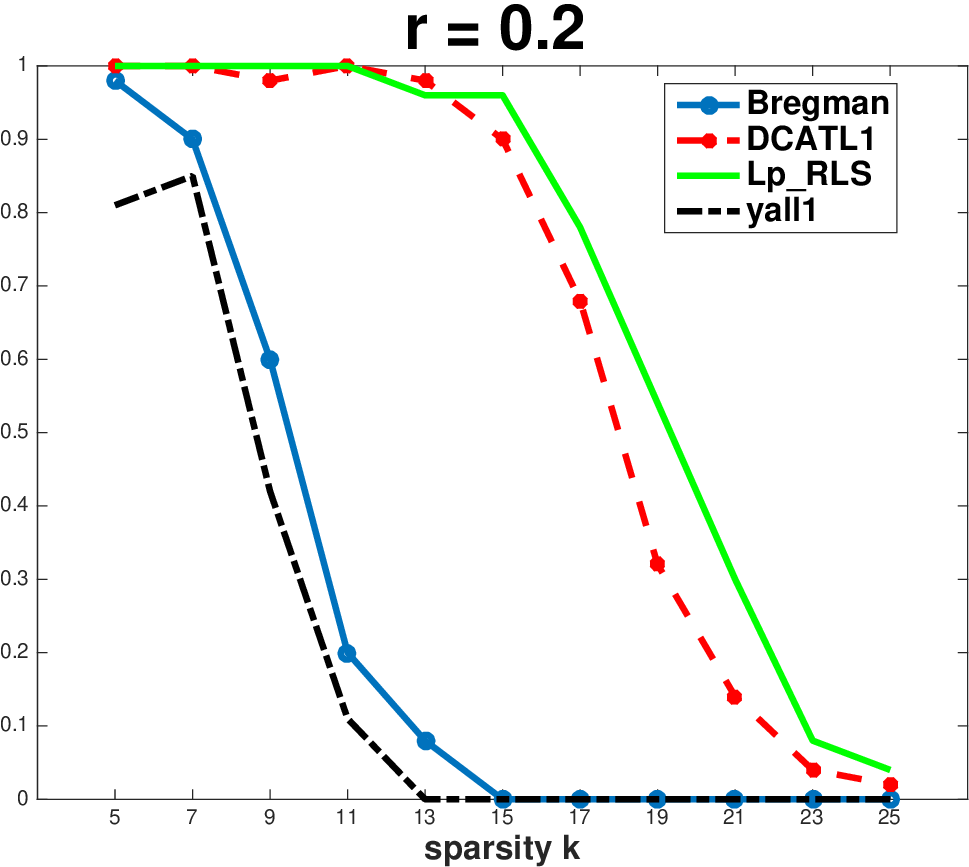}
\end{minipage} \\
\begin{minipage}[t]{0.5\linewidth}
\includegraphics[scale=0.35]{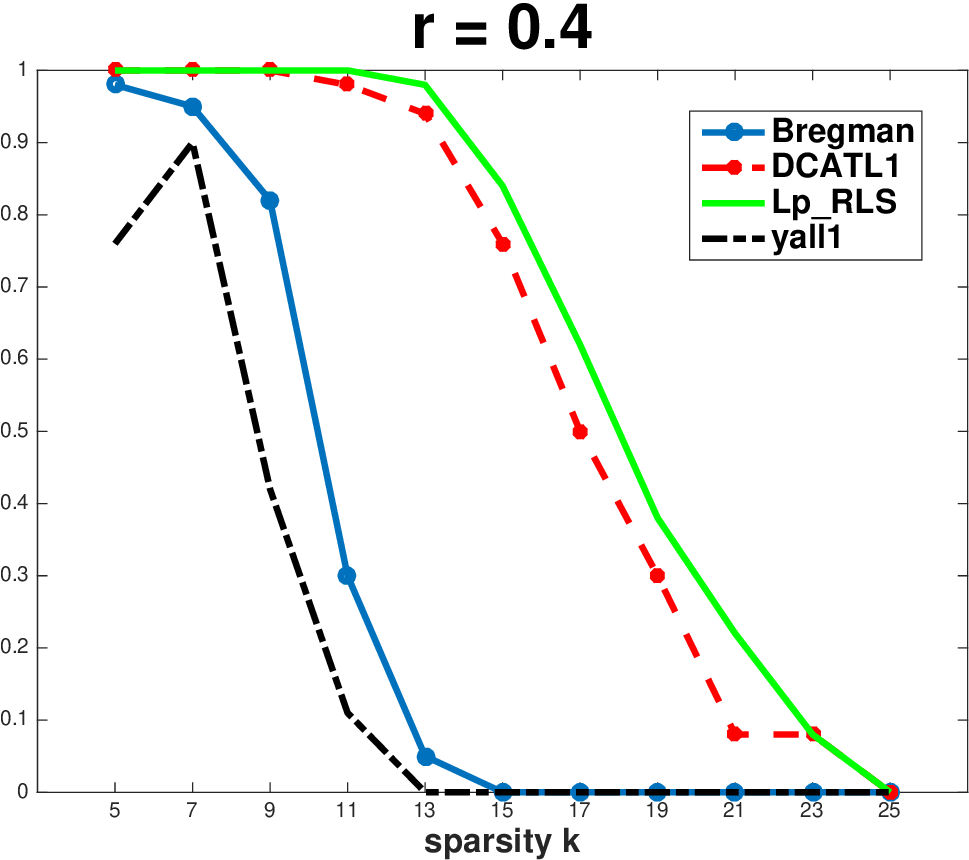}
\end{minipage}  &
\begin{minipage}[t]{0.5\linewidth}
\includegraphics[scale= 0.35]{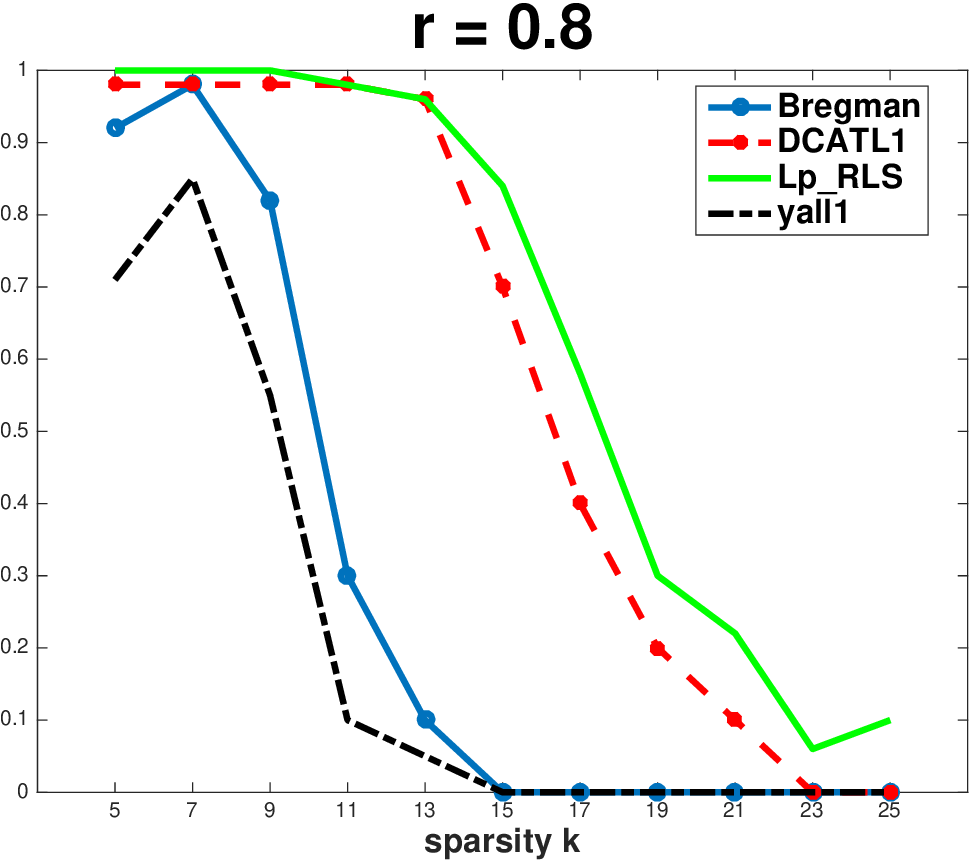}
\end{minipage}  \\
\end{tabular}
\caption{Comparison of constrained algorithms for $64\times 1024$ Gaussian random matrices 
with different coherence parameter $r$. The data points are averaged over 50 trials.}
\label{graph:cons-gaussian}
\end{figure}

\subsubsection{Over-sampled DCT} 
We fix $(M,N)=(100,1500)$, and vary parameter $F$ from 2 to 20, 
so the coherence of there matrices has a wider range and almost reaches 1 at the high end. 
In Fig. (\ref{graph:cons over}), when $F$ is small, say 
$ F = 2,4$, $Lp-RLS$ still performs the best, similar to the case of 
Gaussian matrices. However, with increasing $F$, the success rates for $Lp-RLS$ 
declines quickly, worse than the Bregman $l_1$ algorithm at $F=6,10$.
The performance for DCATL1 is very stable and maintains a high level 
consistently even at the very high end of coherence ($F=20$).     

\begin{figure}
\begin{tabular}{lll}
\begin{minipage}[t]{0.32\linewidth}
\includegraphics[scale=0.26]{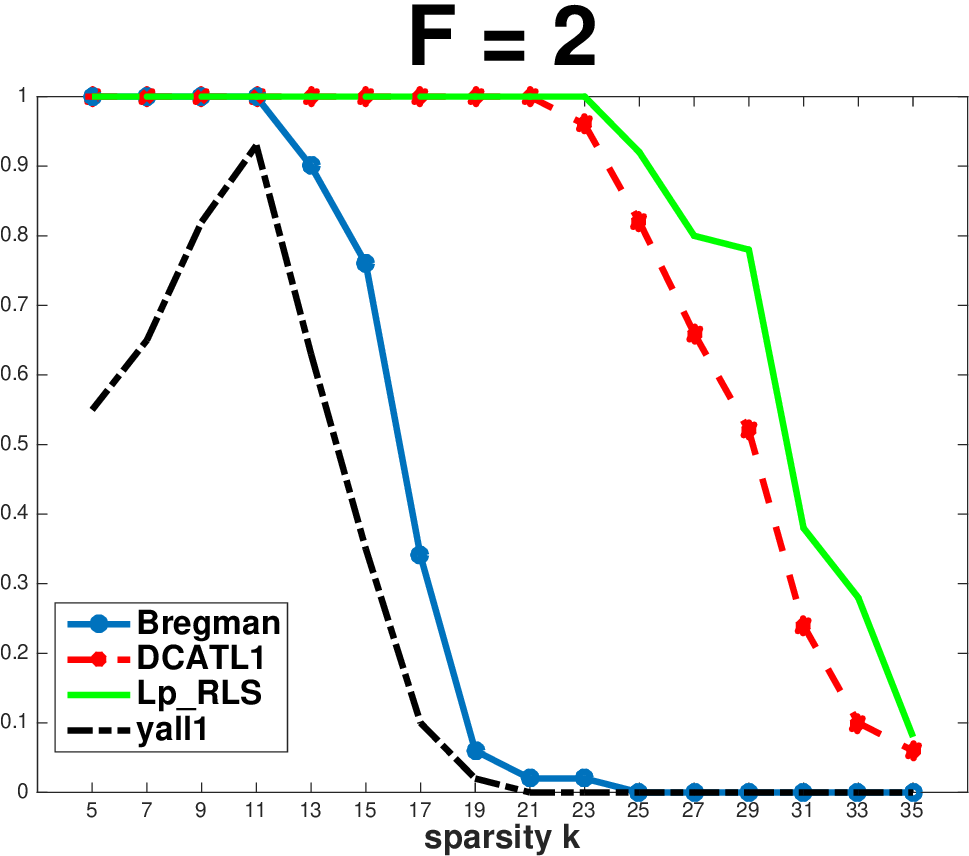}
\end{minipage} &
\begin{minipage}[t]{0.32\linewidth}
\includegraphics[scale= 0.26]{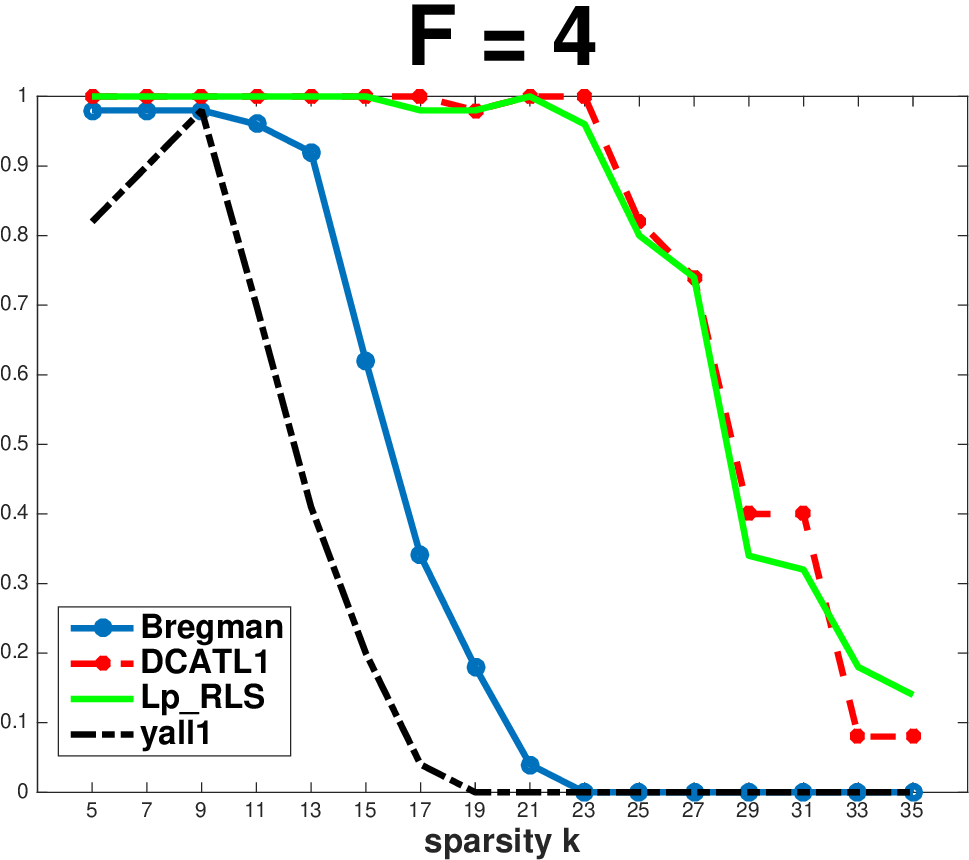}
\end{minipage} &
\begin{minipage}[t]{0.34\linewidth}
\includegraphics[scale= 0.26]{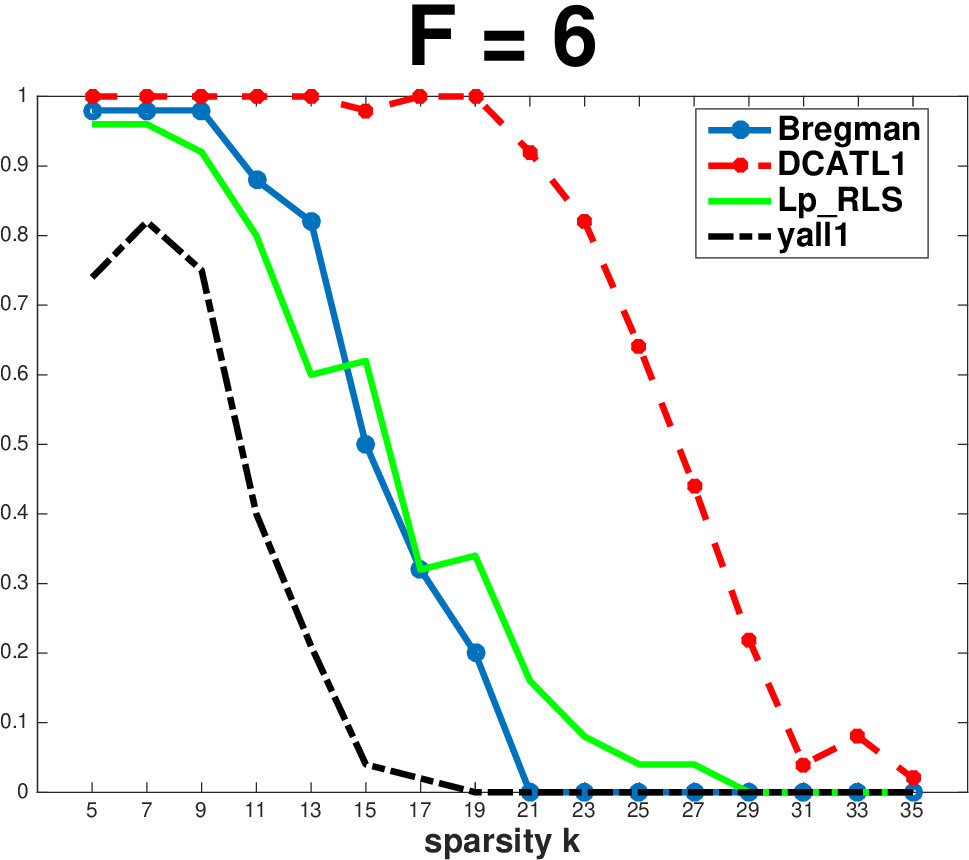}
\end{minipage} \\
\begin{minipage}[t]{0.32\linewidth}
\includegraphics[scale=0.26]{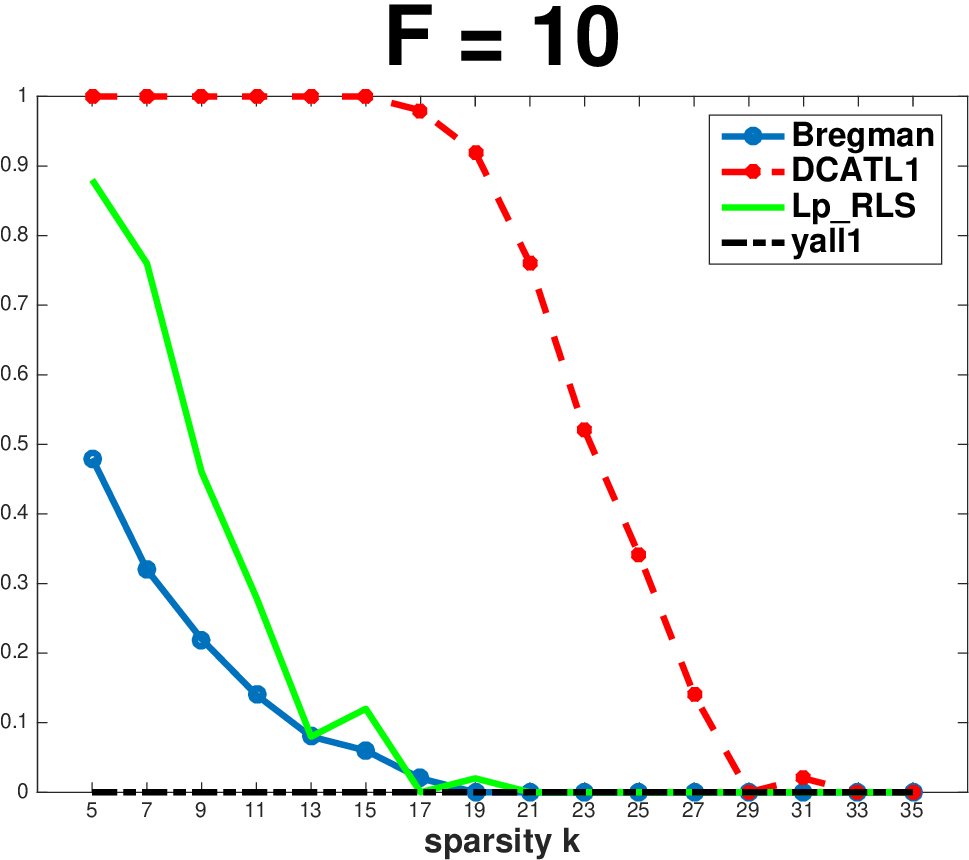}
\end{minipage} &
\begin{minipage}[t]{0.32\linewidth}
\includegraphics[scale= 0.26]{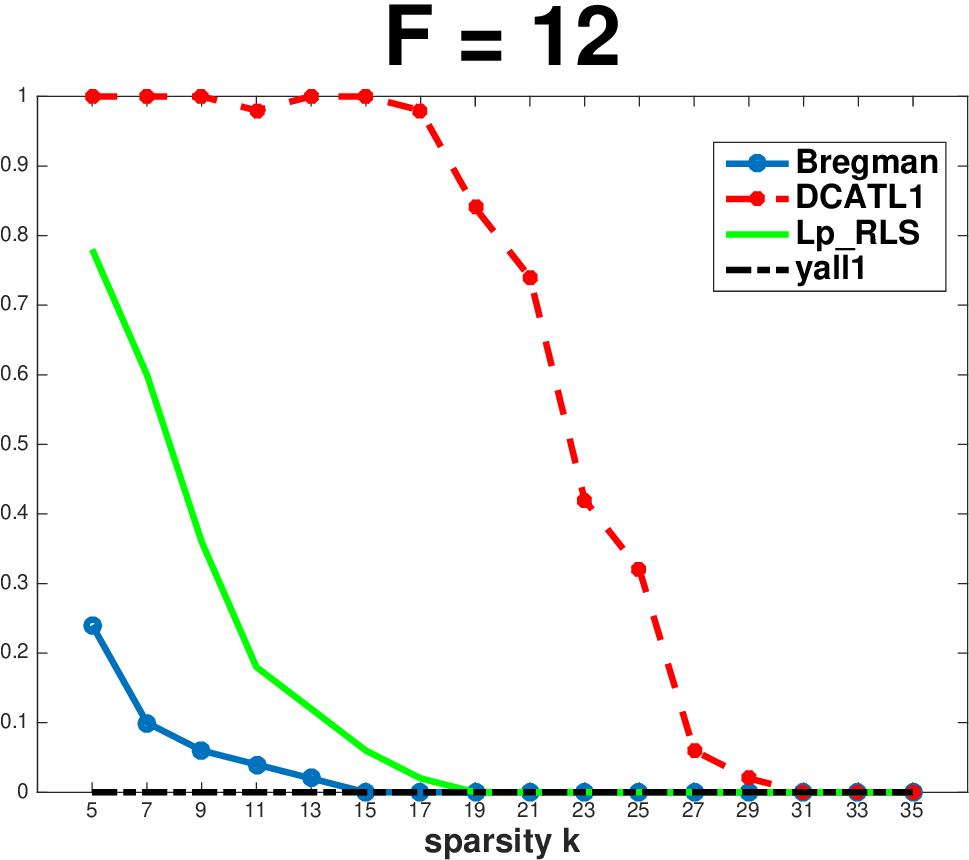}
\end{minipage}  &
\begin{minipage}[t]{0.34\linewidth}
\includegraphics[scale=0.26]{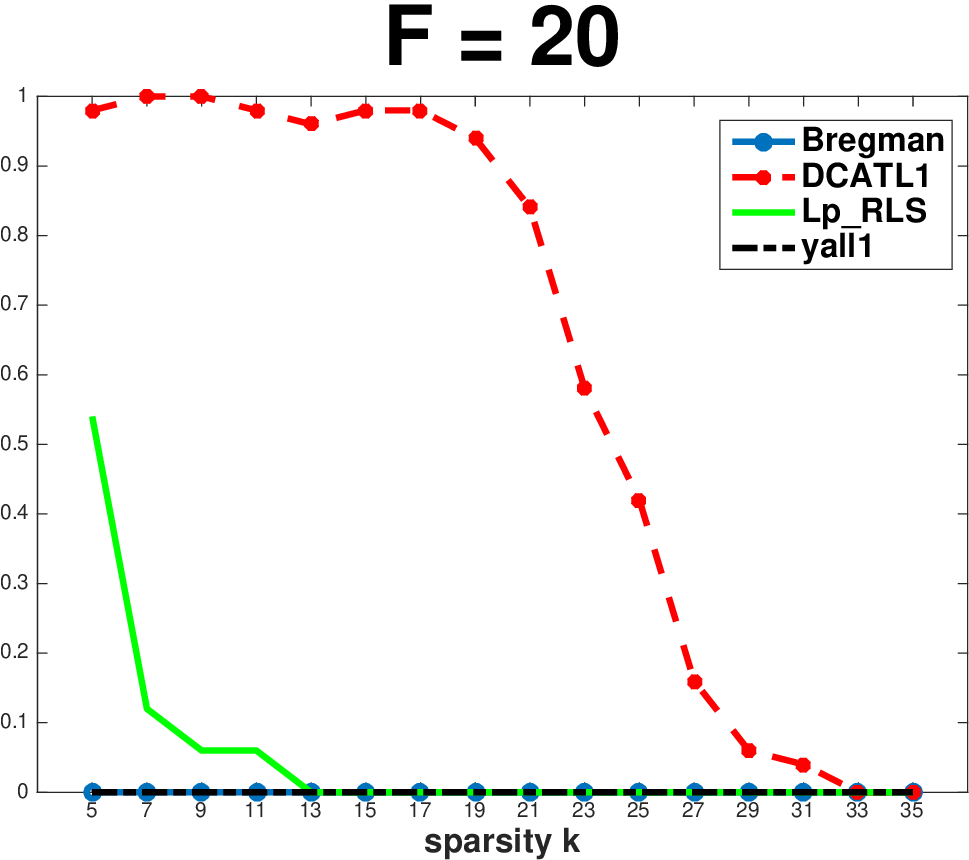}
\end{minipage} \\
\end{tabular}
\caption{Comparison of success rates of constrained algorithms for the over-sampled DCT random matrices: $(M,N) = (100,1500)$ 
with different $F$ values, peak separation by $2RL = 2F$.}
\label{graph:cons over}
\end{figure}
\medskip

\begin{rmk}
In view of evaluation results in this section (Fig. \ref{figure:Gaussian},
Fig. \ref{figure:F}, Fig. \ref{graph:cons-gaussian}, 
Fig. \ref{graph:cons over}), 
we see that DCATL1 algorithms offer robust solutions in CS problems 
for random sensing matrices with a broad range of coherence. In applications where 
sensing hardwares cannot be modified or upgraded, 
a robust recovery algorithm is a valuable tool for information retrieval. 
An example is super-resolution where sparse signals are recovered from 
low frequency measurements within the hardware resolution limit \cite{Super:candes2013mini,LYX_16}. 
\end{rmk}   

\section{Comparison of DCA on Different Non-convex Penalties}
\medskip

In this section, we compare DCA on 
other non-convex penalty functions such as 
PiE \cite{Nguyen-DCA-2015}, MCP \cite{MCP}, 
and SCAD \cite{fan2001variable}.  The computation is based on our 
DCA-ADMM scheme, which uses Algorithm \ref{alg: DCA}
to solve unconstrained optimization and 
Algorithm \ref{alg:ADMM} to solve 
subproblems. The DCA schemes of these penalty functions 
are same as in \cite{Miju, Nguyen-DCA-2015}.

Among the three penalties, PiE has one hyperparameter while MCP and 
SCAD have two hyperparameters. We used $64 \times 256$ Gaussian  
random matries to select best parameters for these penalties. 

\begin{table}
\begin{tabular}{ c| cl r  }
  \hline
 & Formula & Parameters \\
  \hline
	PiE 
		& $\phi(t) = 1 - e^{- \beta |t| }$ 
		& $\beta = 10$  \\  \hline
	MCP
		& $\phi(t) = \alpha \beta^2 - \dfrac{[(\alpha \beta - |t|)_{+}]^2}{\alpha}$ 
		& $\alpha = 5$, $\beta = 0.1$  \\  \hline
	SCAD
		& $ \phi(t) = \begin{cases}
			\beta |t| 
				, & \mathrm{if} \ \ |t| \leq \beta;  \\
             - \dfrac{t^2 - 2 \alpha \beta |t| + \beta^2}{2(\alpha - 1)}  
             	, & \mathrm{if} \ \ \beta < |t| \leq \alpha \beta;   \\
             \dfrac{(\alpha+1) \beta^2}{2} 
             	, & \mathrm{if} \ \ |t| > \alpha \beta.
            \end{cases} $
		& $\alpha = 5$, $\beta = 0.1$  \\  
   \hline
\end{tabular}
\label{table: other non-convex} 
\caption{Three non-convex penalty functions and their parameter values in the numerical experiments.}
\end{table}
\medskip

All other parameters for DCA algorithms during the numerical 
experiments are same as DCA-TL1. The success rate curves 
are shown in Fig. \ref{figure: DCA-compares}. DCA-MCP algorithm
has very good performance on all Gaussian and over-sampled DCT matrices. In all the 
experiments, DCA-TL1 achieves almost the same level of success rates as DCA-MCP.
Consistent with remark \ref{gnsp} and that the set of SCAD gNSP satisfying matrices 
is smaller than those of (PiE,TL1,MCP), SCAD is behind in the two plots of the 
first column of Fig. \ref{figure: DCA-compares} where the sensing matrices are 
in the incoherent regime. Interestingly in the highly coherent regime 
(for over-sampled DCT matrices at $F=20$), DCA-SCAD fares well. 
From PiE behind MCP, TL1 in Fig. \ref{figure: DCA-compares}, 
and $\ell_p$ (in view of Fig. \ref{figure:Gaussian}), the gNSP of PiE is likely to be 
more restrictive than those of MCP, TL1 and $\ell_p$, while the gNSP's of the latter three 
are rather close. A precise characterization of these gNSP's will be interesting for a future work.    
\medskip
  
It is worth pointing out that DCA-TL1 has only one hyperparameter and so 
is easier to adjust and adapt to different tasks. 
Two hyperparameters give more parameter space for  
improvement, but also require more efforts to search for optimal values. 

\begin{figure}
\begin{tabular}{lll}
\begin{minipage}[t]{0.32\linewidth}
\includegraphics[scale=0.24]{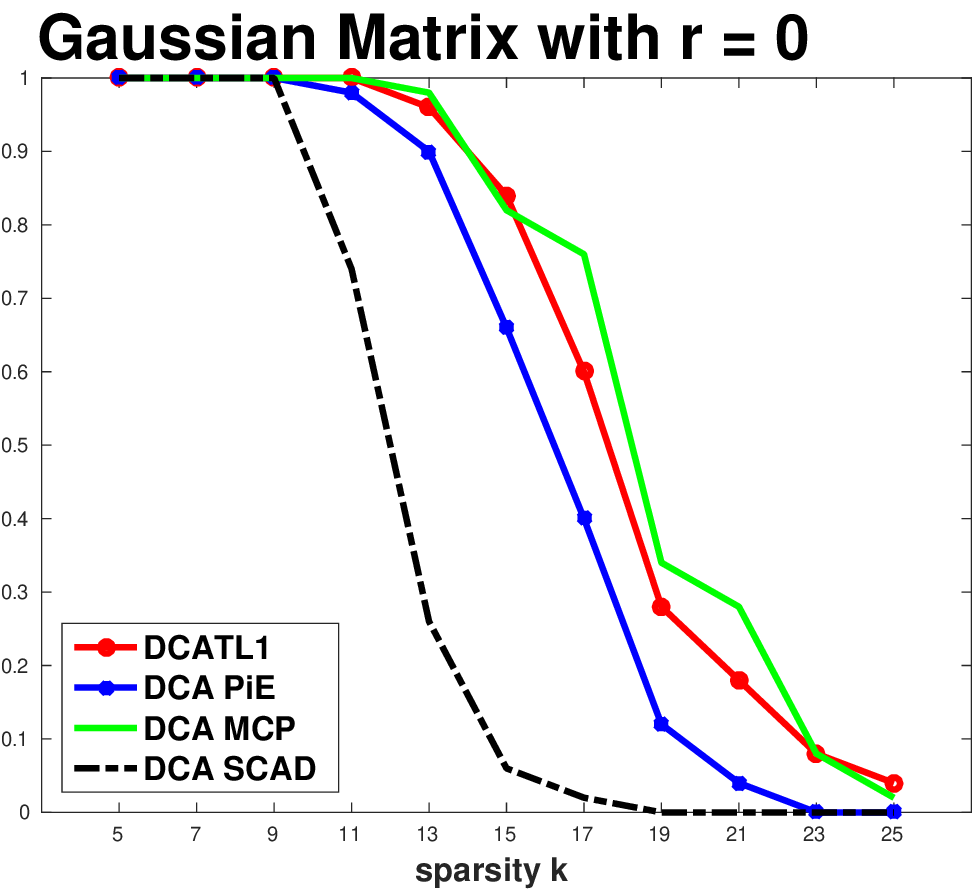}
\end{minipage} &
\begin{minipage}[t]{0.32\linewidth}
\includegraphics[scale= 0.24]{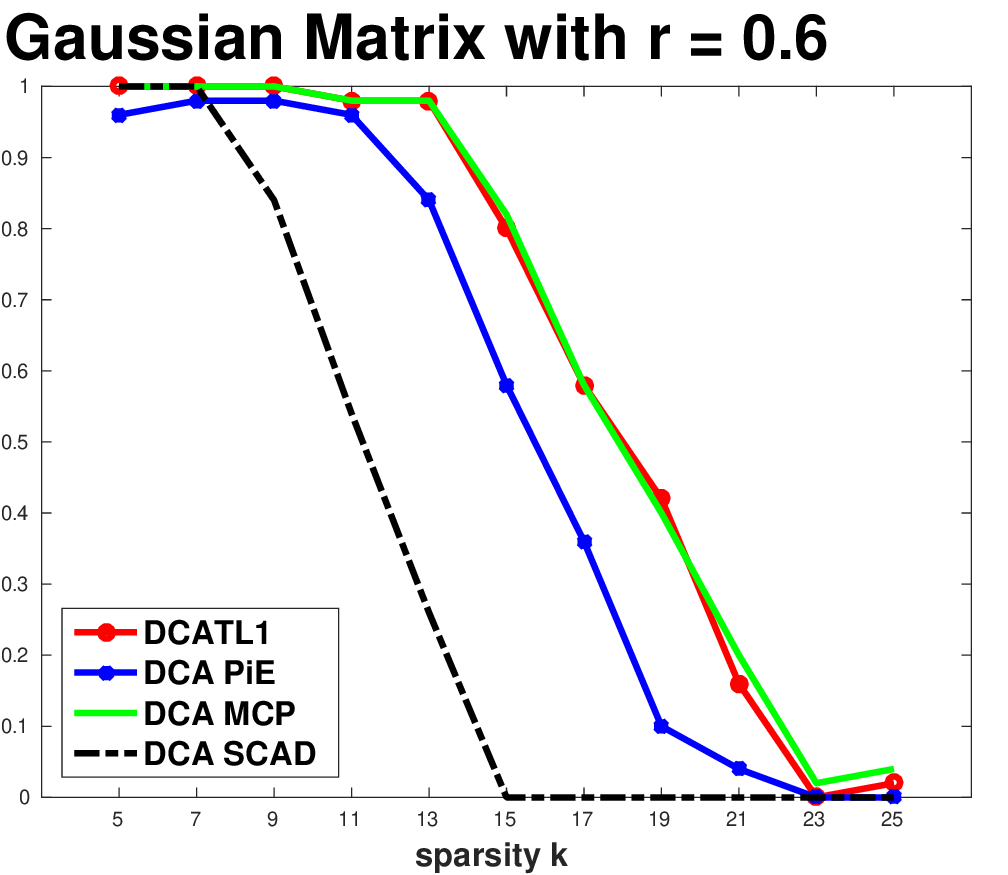}
\end{minipage} &
\begin{minipage}[t]{0.32\linewidth}
\includegraphics[scale=0.24]{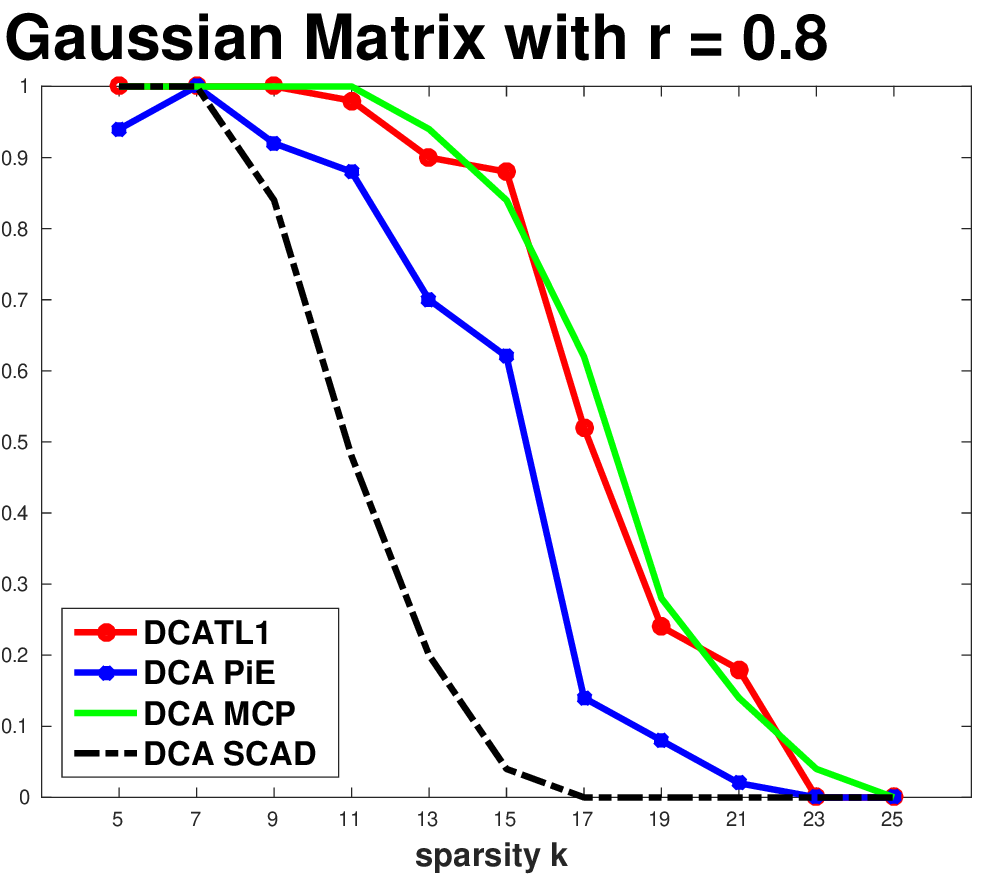}
\end{minipage} \\
\begin{minipage}[t]{0.32\linewidth}
\includegraphics[scale= 0.24]{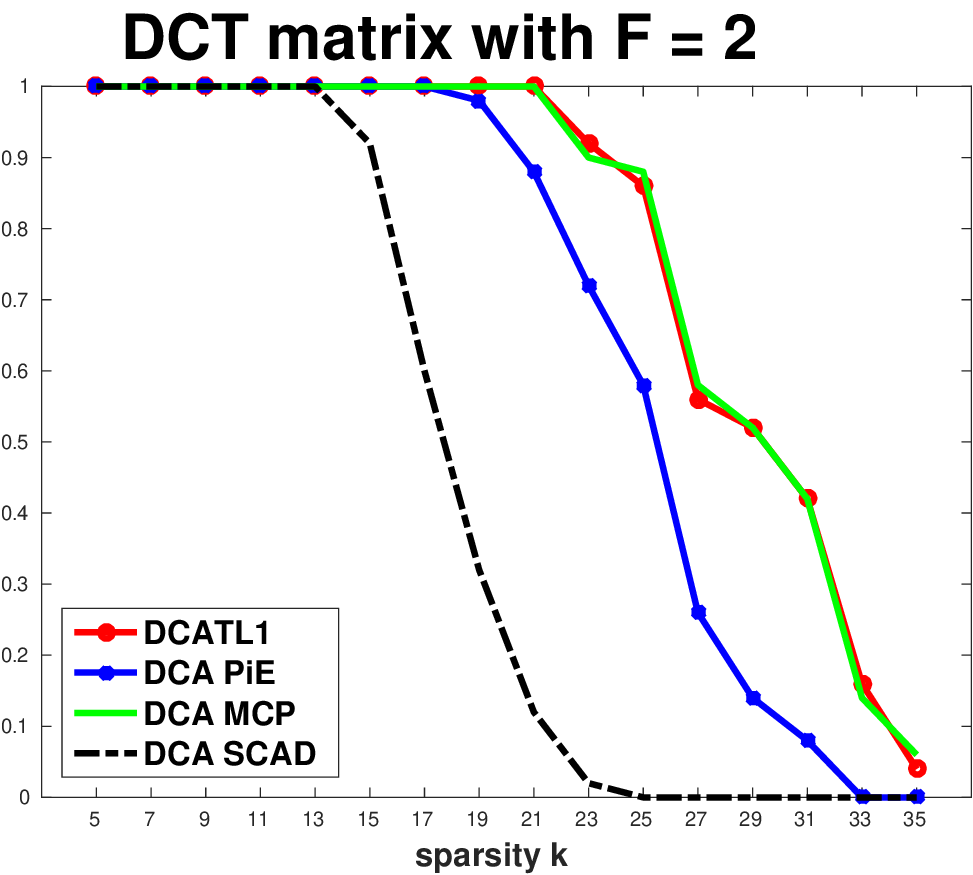}
\end{minipage}  &
\begin{minipage}[t]{0.32\linewidth}
\includegraphics[scale=0.24]{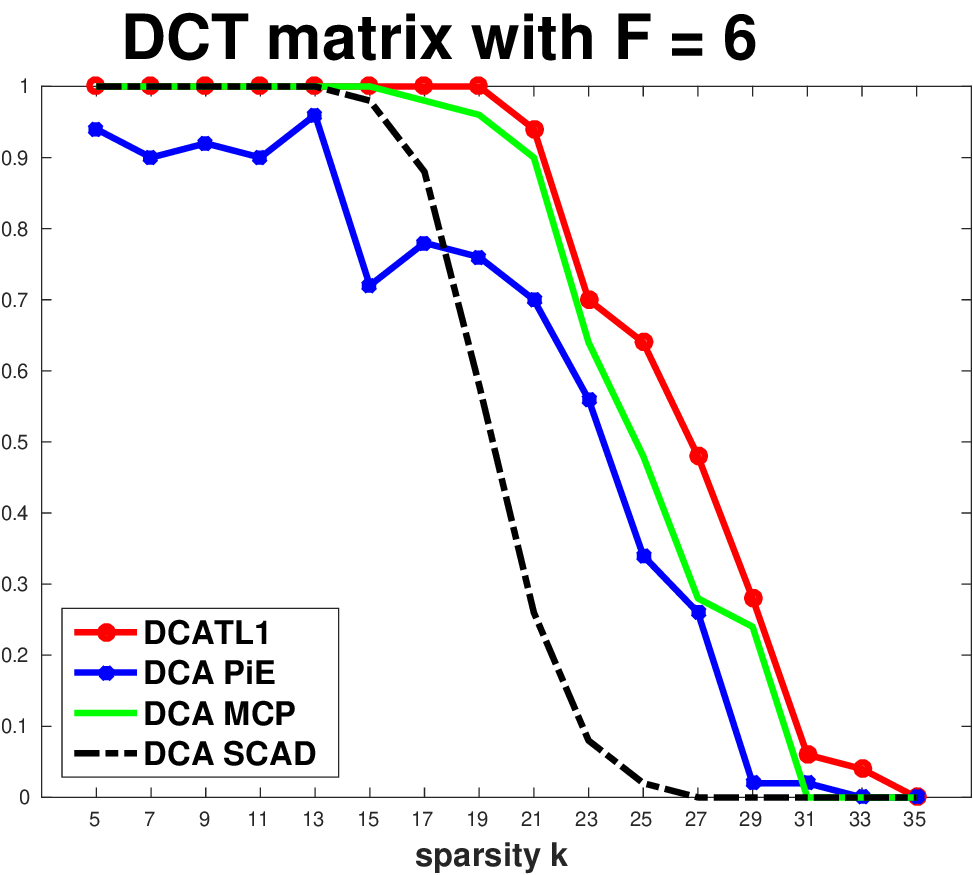}
\end{minipage} &
\begin{minipage}[t]{0.32\linewidth}
\includegraphics[scale= 0.24]{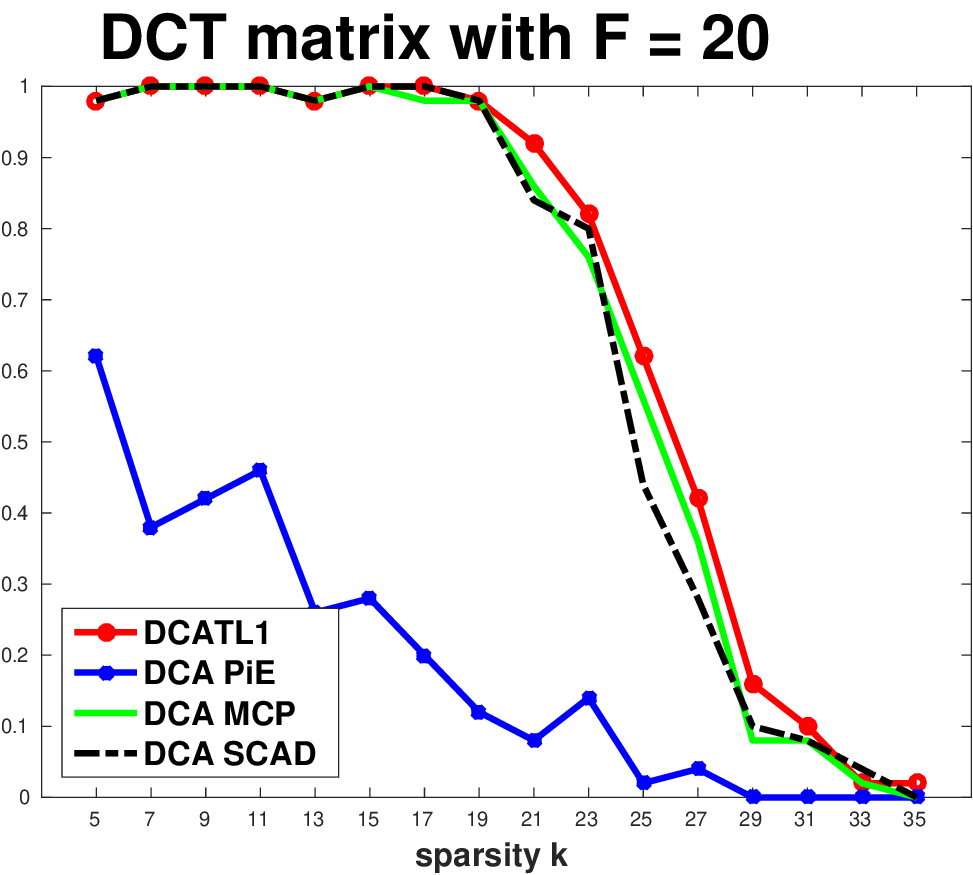}
\end{minipage} 
\end{tabular}
\caption{Numerical tests for DCA-ADMM scheme with different non-convex penalties 
using $64\times 1024$ random Gaussian matrices and $100 \times 1500$
over-sampled DCT matrices of varying $(r,F)$ values.}
\label{figure: DCA-compares}
\end{figure}

\section{Concluding Remarks}
\medskip

We have studied compressed sensing problems with the transformed $\ell_1$ penalty 
function (TL1) for both unconstrained and constrained models with 
random sensing matrices of a broad range of coherence. 
We discussed exact and stable recovery properties of TL1 using null space property 
and restricted isometry property of sensing matrices.  
We showed two DC algorithms along with a convergence theory. 
\medskip 

In numerical experiments, DCATL1 with ADMM solving the convex sub-problems 
is on par with the best method reweighted $l_{1/2}$ ($Lp-RLS$) in the 
unconstrained (constrained) model, in case  
of incoherent Gaussian sensing matrices.
For highly coherent over-sampled DCT random matrices, 
DCATL1 with ADMM solving the convex sub-problems 
is also comparable with the best method DCA $\ell_1-\ell_2$ algorithm.   
For random matrices of varying degree of coherence,   
the DCATL1 algorithm is the most robust for constrained 
and unconstrained models alike. We tested DCA with ADMM inner solver on other 
non-convex penalties (PiE, SCAD, MCP) with one and two hyperparameters, and  
found DCATL1 to be competitive as well. 
\medskip

In future work, we plan to develop TL1 algorithms for 
image processing and machine learning applications.

\section*{Acknowledgments}
The authors would like to thank Professor Wenjiang Fu for 
referring us to \cite{transformed-l1}, 
Professor Jong-Shi Pang for his helpful suggestions,  
and the anonymous referees for their constructive comments.
\medskip

\appendix
\section{Proof of Exact TL1 Sparse Recovery (Theorem 2.1)}
\begin{proof} 
The proof is along the lines of arguments in 
\cite{candes2005error} and \cite{chartrand2007nonconvex}, while using special 
properties of the penalty function $\rho_a$.  
For simplicity, we denote $\beta_{C}$ by $\beta$ and $\beta^0_{C}$ by $\beta^0$. \\

Let $e = \beta - \beta^0$, and we want to prove that the vector $e =0$. 
It is clear that, $e_{T^c} = \beta_{T^c}$, since $T$ is the 
support set of $\beta^0$. By the triangular inequality of $\rho_a$, we have: 
\begin{equation*}
   P_a(\beta^0) - P_a(e_T) = P_a(\beta^0) - P_a(-e_T) \leq P_a(\beta_T).
\end{equation*}
Then 
\begin{equation*}
\begin{array}{lll}
   P_a(\beta^0) - P_a(e_T) + P_a(e_{T^c}) & \leq & P_a(\beta_T) + P_a(\beta_{T^c}) \\
                                          & = & P_a(\beta) \\
                                          & \leq & P_a(\beta^0) \\
\end{array}
\end{equation*}
It follows that: 
\begin{equation}\label{Ineq:TTc}
   P_a(\beta_{T^c}) = P_a(e_{T^c}) \leq P_a(e_T).
\end{equation}

Now let us arrange the components at $T^c$ in the order of 
decreasing magnitude of $|e|$ and partition into $L$ parts: 
$T^c = T_1 \cup T_2 \cup ... \cup T_L$, where each $T_j$ has $R$ elements (except possibly 
$T_L$ with less). Also denote $T = T_0$ and $T_{01} = T \cup T_1$.  
Since $Ae = A(\beta - \beta^0) = 0$, it follows that 
\begin{equation}\label{Ineq:connector}
\begin{array}{lll}
   0 & = & \|Ae\|_2 \\
     & = & \| A_{T_{01}}e_{T_{01}} + \sum \limits_{j = 2}^L A_{T_j}e_{T_j} \|_2 \\
     & \geq & \| A_{T_{01}}e_{T_{01}} \|_2 - \sum \limits_{j = 2}^L  \| A_{T_j}e_{T_j} \|_2 \\
     & \geq & \sqrt{1-\delta_{|T|+R}} \|e_{T_{01}}\|_2 - \sqrt{1+\delta_{R}} 
     \sum \limits_{j = 2}^L \|e_{T_{j}}\|_2
\end{array}
\end{equation}

At the next step, we derive two inequalities between the $l_2$ norm and 
function $P_a$, in order to use the inequality (\ref{Ineq:TTc}). Since

\begin{equation*}
\begin{array}{lll}
   \rho_a(|t|) = \dfrac{(a+1)|t|}{a+|t|} & \leq & (\dfrac{a+1}{a})|t| \\
                           & = & (1 + \dfrac{1}{a})|t|  \\
\end{array}
\end{equation*}
we have: 
\begin{equation} \label{Ineq:2>a}
\begin{array}{lll}
   P_a(e_{T_{0}}) & = &  \sum \limits_{i \in T_{0}} \rho_a(|e_i|) \\
                   & \leq & (1 + \frac{1}{a}) \| e_{T_{0}} \|_1 \\
                   & \leq & (1 + \frac{1}{a}) \sqrt{|T|} \ \ \| e_{T_{0}} \|_2 \\ 
                   & \leq & (1 + \frac{1}{a}) \sqrt{|T|} \ \ \| e_{T_{01}} \|_2. \\
\end{array}
\end{equation}

Now we estimate the $l_2$ norm of $e_{T_j}$ from above in terms of $P_a$. 
It follows from $\beta$ being the minimizer of the problem (\ref{eq:spar revised}) and 
the definition of $x_C$ (\ref{para: scaled pa}) that 
\begin{equation*}
   P_a(\beta_{T^c})  \leq  P_a(\beta) 
                     \leq  P_a(x_{C}) 
                     \leq  1. 
\end{equation*}

For each $i \in T^c,\ \ \rho_a(\beta_i) \leq P_a(\beta_{T^c}) \leq 1$. Also since 
\begin{equation} \label{ineq: rho vs 1}
\begin{array}{ll}
                   & \dfrac{(a+1)|\beta_i|}{a+|\beta_i|} \leq 1 \\ 
   \Leftrightarrow & (a+1)|\beta_i| \leq a+|\beta_i| \\
   \Leftrightarrow & |\beta_i| \leq 1 \\
\end{array}
\end{equation}
we have
\begin{equation*}
   |e_i| = |\beta_i|   \leq  \dfrac{(a+1)|\beta_i|}{a+|\beta_i|} 
                       =  \rho_a(|\beta_i|) \ \ \  \mbox{for every $i \in T^c$}.  
\end{equation*}
It is known that function $\rho_a(t)$ is increasing for non-negative variable $t \geq 0$, and
\[
	|e_i| \leq |e_k|  \  \textit{ for } \   \forall \ i \in T_j 
	\   \textit{ and } \   \forall \ k \in T_{j-1},
\]
where $j = 2,3,...,L$. Thus we have
\begin{equation} \label{Ineq:a>2}
\begin{array}{ll}
\vspace{2mm}
   & |e_i| \leq \rho_a(|e_i|) \leq P_a(e_{T_{j-1}}) /R  \\
\vspace{2mm}
   \Rightarrow &  \|e_{T_j}\|_2^2 \leq \dfrac{P_a(e_{T_{j-1}})^2}{R} \\
\vspace{2mm}
   \Rightarrow &  \|e_{T_j}\|_2 \leq \dfrac{P_a(e_{T_{j-1}})}{R^{1/2}} \\ 
\vspace{2mm}
   \Rightarrow & \sum \limits_{j=2}^{L} \|e_{T_j}\|_2 \leq \sum \limits_{j=1}^{L} \dfrac{P_a(e_{T_j})}{R^{1/2}} \\
\end{array} 
\end{equation}

Finally, plug (\ref{Ineq:2>a}) and (\ref{Ineq:a>2}) into inequality (\ref{Ineq:connector}) to get:
\begin{equation}
\begin{array}{lll}
\vspace{2mm}
   0 & \geq & \sqrt{1-\delta_{|T|+R}} \dfrac{a}{(a+1)|T|^{1/2}}P_a(e_{T}) - \sqrt{1+\delta_{R}} \dfrac{1}{R^{1/2}} P_a(e_T) \\ 
\vspace{2mm}
     & \geq & \dfrac{P_a(e_T)}{R^{1/2}} \left( \sqrt{1-\delta_{R+|T|}} \dfrac{a}{a+1} \sqrt{\dfrac{R}{|T|}} - \sqrt{1+\delta_{R}} \right) 
\end{array}
\end{equation}

By the RIP condition (\ref{RIP}), the factor
$\sqrt{1-\delta_{R+|T|}} \dfrac{a}{a+1} \sqrt{\dfrac{R}{|T|}} - \sqrt{1+\delta_{R}}$ 
is strictly positive, hence $P_a(e_T) = 0$, and $e_T = 0$. 
Also by inequality (\ref{Ineq:TTc}), $e_{T^c} = 0$. 
We have proved that $\beta_{C} =  \beta^0_{C}$. The equivalence of (\ref{eq:spar revised}) 
and (\ref{eq:l0 revised}) holds. If another vector $\beta$ is the optimal solution 
of (\ref{eq:spar revised}), we can prove that it is also equal to $\beta^0_{C}$, using the same 
procedure. Hence $\beta_{C}$ is unique.

\end{proof}

\section{Proof of Stable TL1 Sparse Recovery (Theorem 2.2)}
\medskip

\begin{proof}
Set $ n = A \beta - y_C $. We use three notations below: 
\begin{itemize}
\item[(i)] $\beta_{C}^n \Rightarrow$ optimal solution for the constrained problem (\ref{eq:spar revised noise});
\item[(ii)] $\beta_{C}   \Rightarrow$ optimal solution for the constrained problem (\ref{eq:spar revised}); 
\item[(iii)] $\beta_{C}^0 \Rightarrow$ optimal solution for the $l_0$ problem (\ref{eq:l0 revised}).
\end{itemize}

Let $T$ be the support set of $\beta_{C}^0$, i.e., 
$T = supp(\beta_{C}^0)$, and vector $e = \beta_{C}^n - \beta_{C}^0$. 
Following the proof of Theorem \ref{thm:global}, we obtain: 
\begin{equation*}
\sum \limits_{j = 2}^{L} \| e_{T_j} \|_2 \leq \sum \limits_{j = 1}^{L} \frac{P_a(e_{T_j})}{R^{1/2}} = \frac{P_a(e_{T^c})}{R^{1/2}}
\end{equation*}
and 
\begin{equation*}
   \|e_{T_{01}} \|_2 \geq \frac{a}{(a+1) \sqrt{|T|}} P_a(e_T).
\end{equation*}

Further, due to the inequality $P_a(\beta_{T^c}^n) = P_a(e_{T^c}) \leq P_a(e_T)$ 
from (\ref{Ineq:TTc}) and inequalities in (\ref{Ineq:connector}), we get
\begin{equation*}
   \| Ae \|_2 \geq \frac{P_a(e_T)}{R^{1/2}} C_{\delta}, 
\end{equation*}
where $C_{\delta} = \sqrt{1-\delta_{R+|T|}} \dfrac{a}{a+1} \sqrt{\dfrac{R}{|T|}} - \sqrt{1+\delta_{R}}$.

By the initial assumption on the size of observation noise, we have  
\begin{equation} \label{ineq: A e vs tau}
\| Ae \|_2 =  \| A\beta_{C}^n - A \beta_{C}^0 \|_2 
            =  \| n \|_2 \leq \tau,
\end{equation}
so we have: $P_a(e_T) \leq \dfrac{\tau R_{1/2}}{C_{\delta}}$.

On the other hand, we know that $P_a(\beta_{C}) \leq 1 $ 
and $\beta_{C}$ is in the feasible set of the 
problem (\ref{eq:spar revised noise}). Thus we have the inequality: 
$P_a(\beta_{C}^n) \leq P_a(\beta_{C}) \leq 1$. 
By (\ref{ineq: rho vs 1}), $\beta_{C,i}^n \leq 1$ for each $i$. 
So, we have
\begin{equation}
   |\beta_{C,i}^n| \leq \rho_a(|\beta_{C,i}^n |).
\end{equation}
It follows that 
\begin{equation*}
\begin{array}{lll}
   \vspace{1mm}
   \|e\|_2 & \leq & \|e_T\|_2 + \|e_{T^c}\|_2 = \|e_T\|_2 + \|\beta_{C,T^c}^n\|_2 \\
   \vspace{1mm}
           & \leq & \dfrac{ \|A_T e_T\|_2 }{\sqrt{1-\delta_T}} + \| \beta_{C,T^c}^n \|_1 \\
           \vspace{1mm}
           & \leq & \dfrac{ \|A_T e_T\|_2 }{\sqrt{1-\delta_T}} + P_a(\beta_{C,T^c}^n) 
           = \dfrac{ \|A_T e_T\|_2 }{\sqrt{1-\delta_T}} + P_a(e_{T^c}) \\
           & \leq & \dfrac{\tau}{\sqrt{1-\delta_R}} + P_a(e_T) \leq D\tau,  
\end{array}
\end{equation*}
for a positive constant $D$ depending only on $\delta_R$ and $\delta_{R+|T|}$. 
The second inequality uses the definition of RIP, 
while the first inequality in the last row comes from (\ref{ineq: A e vs tau}) and 
(\ref{Ineq:TTc}). 
\end{proof}




%

\end{document}